\documentclass[lettersize,journal]{IEEEtran}

\IEEEoverridecommandlockouts

\usepackage{hyperref}

\usepackage[font=small]{subcaption} % Required for subfigure
\usepackage[font=small]{caption}
\usepackage{float}
\usepackage{enumitem}
\usepackage{balance}
\usepackage{cite}
\usepackage{amsmath,amssymb,amsfonts}
\usepackage{bbm}
\usepackage{url}
\usepackage{amsthm}
\usepackage{graphicx}
\usepackage{xcolor}
\usepackage{stfloats}
\usepackage{amsthm}
\usepackage{amsmath}

\makeatletter%
\if@twocolumn%
\newcommand{\Figwidth}{\columnwidth}%
\else% \@twocolumnfalse
\newcommand{\Figwidth}{4.5in}%
\fi%
\makeatother%

\allowdisplaybreaks

\renewcommand{\Figwidth}{3.5in}

\def\log{{\rm{log}}}
\def\min{{\rm{min}}}

\def\bD{\mathbf{D}}

\def\bh{\mathbf{h}}
\def\bn{\mathbf{n}}
\def\bw{\mathbf{w}}

\def\by{\mathbf{y}}

\def\bX{\mathbf{X}}
\def\bx{\mathbf{x}}

\def\bu{\mathbf{u}}
\def\bI{\mathbf{I}}

\def\bf{\mathbf{f}}

\def\bepsilon{\boldsymbol{\epsilon}}

\def\btheta{\boldsymbol{\theta}}
\def\bzeta{\boldsymbol{\zeta}}

\def\bbeta{\boldsymbol{\beta}}

\def\trace{\text{tr}}

\newtheorem{remark}{Remark}
\newtheorem{theorem}{Theorem}

\DeclareFontFamily{U}{mathx}{\hyphenchar\font45}
\DeclareFontShape{U}{mathx}{m}{n}{
      <5> <6> <7> <8> <9> <10>
      <10.95> <12> <14.4> <17.28> <20.74> <24.88>
      mathx10
      }{}
\DeclareSymbolFont{mathx}{U}{mathx}{m}{n}
\DeclareFontSubstitution{U}{mathx}{m}{n}
\DeclareMathAccent{\widecheck}{0}{mathx}{"71}
\DeclareMathAccent{\wideparen}{0}{mathx}{"75}

\usepackage{balance}

\begin{document}

% \title{Edge Learning over Wireless Channels with Heterogeneous Fading Dynamics}
% \title{Communication-Efficient Federated Learning under Heterogeneous Fading Dynamics}
\title{Coherence-Aware Over-the-Air Distributed Learning under Heterogeneous Link Impairments}
% \title{learning OTA: pilot reuse under het. fading dynamics}
\author{Mehdi~Karbalayghareh,~\IEEEmembership{Member,~IEEE,} David~J.~Love,~\IEEEmembership{Fellow,~IEEE}, \\ and~Christopher~G.~Brinton,~\IEEEmembership{Senior Member,~IEEE}

\thanks{This work has been accepted for publication in IEEE Journal on Special Areas in Information Theory (JSAIT).}

\thanks{M. Karbalayghareh, D. J. Love, and C. G. Brinton are with the Department of Electrical and Computer Engineering, Purdue University, West Lafayette, IN 47907, USA (e-mail: mkarbala@purdue.edu; djlove@purdue.edu; cgb@purdue.edu). This work was supported in part by the Office of Naval Research (ONR) under grant N00014-21-1-2472 and by the National Science Foundation (NSF) under grants CNS2212565, CNS2225578, and CPS2313109.}

\thanks{\textcolor{black}{A preliminary version of this work has been accepted to the 2026 IEEE INFOCOM conference.}}

}

\maketitle

\begin{abstract}
Distributed machine learning (ML) over wireless networks hinges on accurate channel state information (CSI) and efficient exchange of high-dimensional model updates. These demands are governed by channel coherence time and bandwidth, which vary across devices (links) due to heterogeneous mobility and scattering, causing degraded downlink delivery and distorted uplink over-the-air (OTA) aggregation. We propose a coherence-aware federated learning (FL) framework that jointly addresses impairments on downlink and uplink with communication-efficient strategies. In the downlink, we employ product superposition to multiplex global model symbols for long-coherence (static) devices onto the pilot tones required by short-coherence (dynamic) devices for channel estimation, turning pilot overhead into payload while preserving estimation fidelity. In the proposed scheme, an orthogonal frequency-division multiplexing (OFDM) super-block is partitioned into sub-blocks aligned with the smallest coherence time and bandwidth, enabling consistent channel estimation and stabilizing OTA aggregation across heterogeneous devices. Partial model reception at dynamic devices is mitigated via previous local model filling (PLMF), which reuses prior updates. We establish convergence guarantees under heterogeneous link impairments, imperfect CSI, and aggregation noise. The proposed framework enables efficient scheduling under coherence heterogeneity; analysis and experiments demonstrate notable gains in communication efficiency, latency, and learning accuracy over conventional FL baselines.

% In the uplink, we develop robust aggregation methods to mitigate inconsistencies arising from partial model reception at dynamic devices. 
\end{abstract}


%\blfootnote{The work of}
\IEEEpeerreviewmaketitle
% no \IEEEPARstart
\section{Introduction}
\label{sec:introduction}
\color{black}
The forthcoming sixth generation (6G) of wireless networks is envisioned to be AI-native, where distributed intelligence, edge learning, and task-oriented communications are tightly integrated into the radio access network~\cite{Saad-6G-2020,Giordani-6G-ComMag2020,Andrews-6G,Brinton-6G-ComMag2025}. At the heart of this vision lies the need to support large-scale collaborative machine learning (ML) across heterogeneous devices under strict latency, reliability, and spectral efficiency constraints. Federated learning (FL) has emerged as a canonical paradigm for such distributed training, enabling devices to learn global models without exchanging raw data~\cite{pmlr-v54-mcmahan17a,Li_FedSurvey_SPM2020,Nguyen_JSAC2021}. Although FL is promising in terms of privacy and scalability, it often encounters significant communication constraints over wireless networks~\cite{Ahmed_INFOCOM25_ComEfficient}. Its efficacy hinges on the availability of a reliable communication system and the ability to accurately acquire and exchange link qualities (channel state information, CSI) among nodes, as well as the efficient exchange of high-dimensional model updates.
\color{black}

The 6G roadmap highlights diverse verticals, industrial automation, connected and automated mobility, immersive extended reality (XR), e-health, energy/smart grid, and smart cities, each with stringent latency, reliability, and scale requirements~\cite{NGA_Verticals_2023}. Many of these verticals depend on learning-enabled capabilities (on-device and in-network) to adapt models, schedule tasks, and close control loops in real time, which elevates the need for communication-efficient, learning-aware radio design that can serve heterogeneous devices/links with non-identical mobility and scattering conditions. Therefore, standardization is moving toward networks designed to support learning. Recent 3GPP efforts outline mechanisms for (Artificial Intelligence) AI/ML data collection and model management/delivery across the system, enabling learning-centric operation in 5G-Advanced and beyond~\cite{3GPP-AIML-2025}. In parallel, 5G-Advanced introduces features that bridge toward 6G AI-assisted air-interface functions, mobility/broadband enhancements, and energy optimization, strengthening the substrate on which distributed learning can run at scale~\cite{Lin-3gpp2025}.

These distributed learning objectives over wireless networks critically depend on the availability of CSI, which is governed by the links' \emph{coherence time} and \emph{coherence bandwidth}. While most existing FL frameworks assume uniform fading rates (i.e., equal coherence conditions) across all devices, real-world wireless networks rarely conform to this assumption. Variations in node mobility and scattering environments lead to the {\em coherence disparity}, where devices experience unequal coherence conditions (e.g., the coexistence of low-mobility and high-mobility devices~\cite{Lin_JSAC_HighLowMobility}). This disparity degrades both downlink model delivery and uplink gradient aggregation quality, rendering conventional communication strategies \emph{inefficient} for FL systems.

In the downlink channel estimation phase, a common pilot signal is shared among all receivers, resulting in uniform time and power allocation across devices~\cite{Tong2004SigProcMag,Lozano2010EUWConf}. This is true even when the links have different coherence conditions. When some channels change more quickly in time and/or frequency, the usual pilot sequence  approach may lead to too many or too few resources being devoted to some users. Enforcing strict orthogonality between pilot and data transmission further amplifies this inefficiency, as it increases overhead and reduces the resources available for sending model updates. In fact, under severe coherence disparity, even static devices\footnote{Throughout, we classify devices by channel coherence conditions. 
\emph{Static} devices have long coherence time (low mobility) and/or large coherence bandwidth (frequency-flat), so their channels vary slowly and require infrequent pilot refresh (known channels). 
\emph{Dynamic} devices have short coherence time (high mobility) and/or small coherence bandwidth (frequency-selective), so their channels vary quickly and require frequent pilots.}, which have less trouble participating in FL rounds, may fail to receive the full model due to bandwidth wasted on redundant pilots. This can significantly increase bias in the learned parameters and degrade overall FL performance.

To address this challenge, a \emph{coherence-aware} design together with \emph{pilot reuse} techniques must be integrated into the FL framework to ensure both bandwidth and resource efficiency, efficient device scheduling, as well as robust model delivery and aggregation. In the downlink, pilot reuse through \emph{product superposition} has emerged as a particularly effective technique under coherence disparity~\cite{Mehdi_TWC24_CoherenceDis,Fadel2016CoherenceDisparity}. This approach overlays data for slow-fading (static) users onto pilot symbols intended for fast-fading (dynamic) users, enabling simultaneous pilot and data transmission within the same time-frequency frame. As a result, fast users obtain fresh pilots as often as needed, while static users exploit unused pilot capacity to receive model parameters at minimal additional cost. Coupling this strategy with coherence-aware device scheduling and uplink model aggregation has the potential to significantly reduce overhead/latency while guaranteeing timely delivery of full model updates and ensuring that all devices, including static ones, can remain active participants in the FL process. In support of such designs, the current 4G LTE (Long Term Evolution) and 5G NR (New Radio) already provide configurable reference signals (RS), which allow tuning of their density in both time and frequency, enabling adaptation to link coherence~\cite{TS36.211,TS38.211,TS38.214}. Building on this capability, we design product superposition pilots and optimize pilot and model-delivery resources for learning under coherence disparity.

\subsection{Related Works}
Communication efficiency and reliability have been widely acknowledged as critical bottlenecks in practical wireless FL, as high-dimensional model updates must be exchanged over bandwidth-constrained and noisy channels~\cite{Lim_ComSurvey2020,Hu_ComSurvey2021,Amiri_TWC2021_FLConvergence,Ahn2020FedDistill}. Most existing works have focused on {\em uplink communication} challenges and tried to reduce the overhead from devices to the parameter server (PS), proposing two main techniques. The first is digital FL, where orthogonal resource blocks are allocated to each device so that the PS can decode and aggregate gradients individually. The second is over-the-air federated learning (OTA-FL), which exploits the natural superposition property of the wireless multiple-access channel to perform simultaneous analog transmissions, enabling one-shot gradient aggregation~\cite{Yang2020OTAComp,Amiri2020FLFading,Zhu2020Broadband,Sery2021OTAFLHeterogeneous}. While digital FL emphasizes efficient scheduling and bandwidth allocation~\cite{Shi2021JointScheduling,Yang2020Scheduling,Wang2022QuantizedFL,Bouzinis2023WirelessQFL,Salehi2021FLUnreliable}, OTA-FL relies on power control to mitigate aggregation noise~\cite{Zhu2021OneBitOTA,Michelusi2024NonCoherent,Wang_INFOCOM25_delayedCSI, Abarghouyi_TWC2024}. However, both lines of work typically assume relatively homogeneous wireless conditions—not only in terms of path loss, but also in terms of channel coherence conditions, where all devices are presumed to experience similar fading dynamics. 

% Communication efficiency and reliability have been widely acknowledged as critical bottlenecks in practical wireless FL, as high-dimensional model updates must be exchanged over bandwidth-constrained and noisy wireless channels~\cite{Lim_ComSurvey2020,Hu_ComSurvey2021,Amiri_TWC2021_FLConvergence,Ahn2020FedDistill}. In the FL literature, the majority of studies have primarily focused on {\em uplink communication} imperfection and tried to reduce the overhead from devices to the PS, proposing two main techniques. The first is digital FL, which allocates orthogonal resource blocks to each device so that the PS can decode and aggregate local gradients individually. The second is over-the-air (OTA) FL, which utilizes the superposition property of the wireless multiple-access channel to perform simultaneous analog transmissions, enabling one-shot gradient aggregation. While digital FL emphasizes efficient scheduling and bandwidth allocation~\cite{Shi2021JointScheduling,Yang2020Scheduling,Wang2022QuantizedFL,Bouzinis2023WirelessQFL,Salehi2021FLUnreliable}, OTA-FL focuses on power control mechanisms to mitigate aggregation noise~\cite{Yang2020OTAComp,Amiri2020FLFading,Zhu2021OneBitOTA,Zhu2020Broadband,Michelusi2024NonCoherent,Sery2021OTAFLHeterogeneous,Wang_INFOCOM25_delayedCSI}. However, both lines of work typically assume relatively homogeneous wireless conditions—not only in terms of path loss, but also in terms of channel coherence conditions, where all devices are presumed to experience similar fading dynamics. 

There are relatively fewer studies focused on imperfect downlink transmission, i.e., for broadcasting the global FL model and its impact on system performance. \textcolor{black}{Azimi-Abarghouyi and Fodor~\cite{Abarghouyi_TWC2024_Hierarchical} proposed a scalable hierarchical OTA-FL framework that employs multi-tier aggregation and a bandwidth-limited downlink broadcast scheme to support learning over clustered wireless networks.} Amiri~{\em et al.}~\cite{Amiri2021NoisyDL} studied the performance of FL over noisy downlink channels, proposing analog (unquantized) and digital (quantized) model broadcasting from the PS to devices with imperfect CSI used for decoding. Building on this direction, Park~{\em et al.}~\cite{Park2021NoisyFeedback} considered feedback imperfections in the form of noisy and limited-rate links during downlink transmission, analyzing their effect on the convergence of distributed gradient methods. Similarly, Nguyen~{\em et al.}~\cite{Nguyen2021CSIUncertainty} studied the impact of uncertain CSI on downlink transmission, proposing robust aggregation schemes for FL in the presence of imperfect CSI at the devices. Cui~{\em et al.}~\cite{Cui2022Beamforming} addressed downlink imperfections by proposing a joint beamforming strategy at the PS to improve model aggregation quality under fading. Addressing communication efficiency, Caldas~{\em et al.}~\cite{Caldas2018ReducingClients} introduced a system that reduces the downlink communication load by selectively distributing compressed model updates tailored to client resource constraints. Along similar lines, Tang~{\em et al.}~\cite{Tang2019DoubleSqueeze} proposed an error-compensated compression scheme where the downlink model is doubly compressed using stochastic gradient and memory error correction techniques, mitigating the impact of bandwidth constraints on FL performance.

While prior works consider downlink and uplink transmission imperfections in FL networks, they overlook a critical factor: {\em coherence disparity}, where devices experience unequal channel coherence time and bandwidth due to mobility and environmental heterogeneity. In the downlink, this mismatch induces uneven pilot requirements and bandwidth inefficiencies, degrading global model delivery for fast-fading devices and wasting resources for static ones (increasing latency). In the uplink, coherence disparity disrupts OTA aggregation: dynamic devices frequently deliver partial or distorted updates under rapid fading, yielding biased or unevenly weighted aggregates at the PS. In effect, the well-known downlink and uplink challenges of wireless FL become even {\em more severe} under coherence disparity. Consequently, conventional pilot/data transmission, scheduling, and aggregation protocols that assume homogeneous coherence might be highly inefficient in such settings.

% Such mismatches result in uneven pilot requirements and bandwidth inefficiencies, which can degrade global model delivery, especially for fast-fading devices, and waste resources of static devices. 

% In addition to downlink inefficiencies, the uplink phase introduces fundamental challenges for OTA aggregation. Unlike digital transmission, where each device occupies orthogonal resources, OTA aggregation relies on the natural superposition of signals over the wireless multiple-access channel, enabling one-shot gradient averaging~\cite{Yang2020OTAComp,Amiri2020FLFading,Sery2021OTAFLHeterogeneous,Zhu2020Broadband}. This approach, however, is highly sensitive to channel estimation accuracy, synchronization, and fading dynamics~\cite{Michelusi2024NonCoherent,Zhu2021OneBitOTA,Wang_INFOCOM25_delayedCSI}. Under coherence disparity, these issues are further amplified: devices with different coherence block lengths must place pilots at unequal intervals, leading to misaligned channel training structures across the network. As a result, the parameter server receives distorted or unevenly weighted aggregates, where contributions from dynamic devices are corrupted by frequent estimation errors while static devices suffer from redundant pilot overhead. This structural mismatch fundamentally complicates the design of pilot allocation, power control, and aggregation protocols on the uplink, making communication efficiency even more critical to sustain reliable federated learning under heterogeneous coherence conditions.

\subsection{Contributions}
Motivated by this, we study OTA-FL systems under coherence disparity and propose a coherence-aware framework with communication-efficient downlink and uplink designs. We analyze the effectiveness of our approach in enhancing communication efficiency and reliability through careful design of overlapping pilot and parameter transmission in the downlink. To the best of our knowledge, this is the first work to address federated learning under coherence disparity with a joint downlink–uplink treatment. Our scheme and theoretical results offer new insights into the design of wireless FL systems under heterogeneous coherence conditions, a setting that is both practically pervasive and theoretically unexplored.

The main contributions of this paper are as follows:
\begin{itemize}[leftmargin=5mm]
    \item We introduce a coherence-aware FL system model that captures downlink and uplink heterogeneity due to the coexistence of static and dynamic devices with unequal coherence times and/or coherence bandwidth. Our model harmonizes pilot reuse techniques with FL system design, paving the way for more bandwidth-efficient learning under coherence disparity.
    \item We employ product superposition in orthogonal frequency-division multiplexing (OFDM) symbols to enable overlapping pilot and parameter transmission in the downlink. This allows static devices to reuse pilot slots for receiving the global model, while dynamic devices can coherently decode the partial model by estimating their respective {\em virtual channels}, which is the product of their own link gain and the parameter signal intended for static devices.
    \item We propose an efficient coherence-aware device scheduling and adaptive gradient aggregation strategy via uplink MAC (multiple access channel) to address partial model reception, utilizing {\em previous local model filling (PLMF)}, which leverages prior local model entries to compensate for missing parameters at dynamic devices.
    \item We provide a convergence analysis of the proposed scheme under imperfect CSI (both downlink and uplink), capturing the impact of estimation errors and fading mismatch, as well as mismatched model update at different devices on learning performance.
\end{itemize}

% The remainder of the paper is organized as follows: Section~\ref{sec:system-model} introduces the system and channel models, and describes the general coherence disparity in the network. Section~\ref{sec:communication-protocol} develops the communication protocol to efficiently schedule both static and dynamic devices to participate in the FL within a unified framework. Section~\ref{sec:uplink-signal-model} presents the uplink OTA aggregation design under the mismatched coherence blocks. Section~\ref{sec:convergence} provides convergence analysis results. Section~\ref{sec:numerical-results} reports numerical simulations, and Section~\ref{sec:conclusion} provides the concluding remarks.

A preliminary version of this work \cite{Mehdi-infocom26} was limited to frequency-flat channels and perfect uplink aggregation.

% \emph{Notation:} Matrices and vectors are denoted by bold capital letters and bold small letters, respectively.  For a matrix $\bA$, the transpose is denoted by $\bA^T$, and the Hermitian by $\bA^H$. $\bI_m$ denotes the $m \times m$ identity matrix. The statistical expectation is denoted by $\mathbb{E}(\cdot)$. 

\section{System Model}
\label{sec:system-model}
We consider an FL system with a PS and $K$ edge devices, depicted in Fig.~\ref{fig:scenario}. Throughout the paper, we will use the terms ``device'', ``receiver'', and ``user'' interchangeably to denote the edge terminals in the downlink. Each device $k \in [K] = \{1, \ldots, K\}$ possesses a local dataset $\mathcal{B}_k$ with cardinality $B_k = |\mathcal{B}_k|$ datapoints. Let $B \triangleq \sum_{k=1}^{K} B_k$, and $F_k(\boldsymbol{\theta}) \triangleq \frac{1}{B_k} \sum_{\boldsymbol{v} \in \mathcal{B}_k} f(\boldsymbol{\theta}, \boldsymbol{v})$ denote the local loss at device~$k$, where $f$ is the empirical loss function. For a $d$-dimensional global model denoted by $\btheta \in \mathbb{R}^d$, the global loss function to be minimized is
\begin{align}
\label{eq:global-loss}
   F(\btheta) = \sum_{k=1}^{K} \frac{B_k}{B} F_k(\btheta).
\end{align}

FL aims to minimize $F(\btheta)$ through iterative collaboration between the PS and $K$ edge devices. In each iteration~$t$, the PS broadcasts the global model $\btheta^{(t)}$ to the devices over wireless channels, which are subject to impairments such as fading, noise, and decoding errors. Consequently, each device~$k$ receives an imperfect version of the model, denoted by ${\bar{\btheta}}_k^{(t)}$. Each device then performs $\tau$ steps of stochastic gradient descent (SGD) using its local data. At step~$i$, it selects a random minibatch $\bbeta_{k,i}^{(t)}$ and updates its model via
\begin{align}
\btheta_{k,i+1}^{(t)} = \btheta_{k,i}^{(t)} - \eta_{k,i}^{(t)} \nabla F_k \Big(\btheta_{k,i}^{(t)}, \bbeta_{k,i}^{(t)}\Big), \quad i \in [\tau],
\label{eq:local-updates}
\end{align}
where $\btheta_{k,1}^{(t)} = \bar{\btheta}_k^{(t)}$ and $\eta_{k,i}^{(t)}$ is the learning rate. After local updates, device~$k$ sends $\Delta \btheta_k^{(t)} = \btheta_{k,\tau}^{(t)} - \btheta_{k,1}^{(t)}$ to the PS, which receives a noisy estimate $\widehat{\Delta{\btheta}}_k^{(t)}$. The PS aggregates these updates to refine the global model as
\begin{align}
\btheta^{(t+1)} = \btheta^{(t)} + \sum_{k=1}^{K} \frac{B_k}{B} \widehat{\Delta{\btheta}}_k^{(t)}.
\label{eq:model-update}
\end{align}
This process continues until convergence over $t = 1,\ldots,T$ training rounds. The detailed uplink aggregation design is provided in Section~\ref{sec:uplink-signal-model}

\begin{figure}[t]
    \centering
    \begin{subfigure}{\Figwidth}
        \centering
        \includegraphics[page=1, width=0.85\columnwidth]{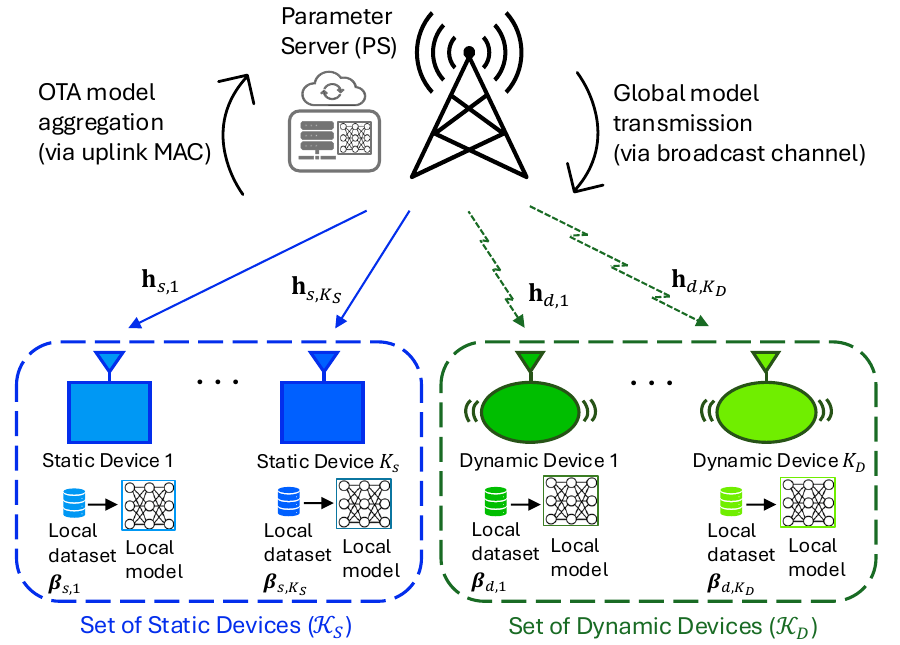}
        \caption{OTA-FL scenario considered.}
    \end{subfigure}

    \vspace{0.3cm}

    \begin{subfigure}{0.75\columnwidth}
        \centering
        \includegraphics[page=2, width=\columnwidth]{Figures/Fig1b-blocks.pdf}
        \caption{Heterogeneous dynamic links.}
    \end{subfigure}

    \caption{(a) The OTA-FL scenario with heterogeneous devices. 
    % Static devices experience large coherence times and bandwidths (low mobility, frequency-flat channels), while dynamic devices experience shorter coherence times and bandwidths (high mobility, frequency-selective channels). 
    (b) Unequal coherence blocks (time--frequency frames) for dynamic devices; different colors denote distinct fading states within each block. Green squares denote the pilot placement within each coherence block to estimate the channel, where pilot--parameter superposition is employed. \vspace{-0.2in}}
    \label{fig:scenario}
\end{figure}

% \begin{figure*}[t]
%     \centering
%     \includegraphics[page=1, width=0.6\textwidth]{Figures/scenario.pdf}
%     \caption{The OTA-FL scenario with heterogeneous devices. Static devices experience large coherence times and bandwidths (low mobility, frequency-flat channels), while dynamic devices experience shorter coherence times and bandwidths (high mobility, frequency-selective channels). As indicated above, dynamic devices have unequal coherence blocks (time-frequency frames); different colors denote distinct fading states within each block.}
%     \label{fig:scenario}
% \end{figure*}

\subsection{Channel Model}
\label{sec:channel-model}
We assume that the PS is equipped with $M$ antennas, while each device has a single antenna. Communication proceeds over an OFDM grid with $N_s$ symbols and $N$ active subcarriers per round, forming the time--frequency index set
\begin{align*}
\mathcal{G} \triangleq \{(n,m): 0\!\le n\!\le \!N_s, 0\!\le m\! \le \!N\}.
\end{align*}
The downlink channel from the PS to device $k$ at symbol $n$ and subcarrier $m$ is $\bh_k[n,m] \in \mathbb{C}^{M\times 1}$ with independent, identically distributed (i.i.d.) entries $\mathcal{CN}(0,1)$. Let $\bx[n,m]\in\mathbb{C}^{M\times 1}$ denote the transmit vector from the PS on $(n,m)$. The received signal at device $k$ is
\begin{align}
    y_k[n,m] = \bh_k[n,m]^H \bx[n,m] + w_k[n,m], \,\, (n,m)\in\mathcal{G},
    \label{eq:received-signal-OFDM}
\end{align}
where $w_k[n,m]\sim\mathcal{CN}(0,\sigma_w^2)$ denotes additive white Gaussian noise (AWGN). The PS satisfies an average power constraint $\rho$ as
\begin{align}
    \frac{1}{N_s N}\mathbb{E}\!\left[\trace\left(\bx[n,m] \bx[n,m]^H\right)\right] \le \rho, \quad (n,m)\in\mathcal{G}.
    \label{eq:PSpower-constraint}
\end{align}

The system operates under the block-fading model with \emph{coherence disparity}: for device~$k$, $\bh_k$ remains constant within a \emph{coherence block} spanning $L_{t,k}\times L_{f,k}$ channel uses, where $L_{t,k}$ and $L_{f,k}$ respectively denote the number of symbols (time slots) and subcarriers, and changes independently across blocks. Due to heterogeneous mobility and scattering, $(L_{t,k}, L_{f,k})$ differ across devices (links)\footnote{Links may be frequency-flat or frequency-selective and either time-invariant or time-varying, depending on delay and Doppler spreads.}. As illustrated in Fig.~\ref{fig:scenario}, we partition devices into a set of static devices $\mathcal{K_S}$ (devices with comparatively larger $L_{t,k}$ and $L_{f,k}$) and a set of dynamic devices $\mathcal{K_D}$ (devices with comparatively smaller $L_{t,k}$ and/or $L_{f,k}$), with $\mathcal{K_S}\cap\mathcal{K_D}=\emptyset$.

% \begin{remark}
% We focus on a frequency-division duplexing (FDD) framework for clarity: one band is used for the downlink (carrying pilots and the model broadcast) and another band for the uplink (client updates). The framework and analyses extend to time-division duplexing (TDD) with minor changes in pilot placement; details are omitted for brevity.
% \end{remark}

\subsection{Pilot Assignment}
\label{sec:pilot-assignment}
Each device~$k$ requires pilot symbols to estimate its link gain within its own coherence block. Due to heterogeneous mobility and scattering, each device may experience a different time--frequency coherence structure, and thus the pilot duty cycle is generally not uniform across all devices.

To facilitate scalable pilot design, pilots are arranged on a separable lattice over the time--frequency grid $\mathcal{G}$, with time density $\lambda_t \in (0,1)$ and frequency density $\lambda_f \in (0,1)$. This means that pilot symbols are inserted every $\frac{1}{\lambda_t}$ OFDM symbols in time and every $\frac{1}{\lambda_f}$ subcarriers in frequency. Let $\mathcal{P} \subset \mathcal{G}$ denote the set of pilot positions and $\mathcal{D} \triangleq \mathcal{G} \setminus \mathcal{P}$ the set of data positions. The overall pilot fraction is given by
\begin{align}
\lambda \triangleq 1 - (1 - \lambda_t)(1 - \lambda_f),
\label{eq:pilot-fraction}
\end{align}
such that $|\mathcal{P}| = \lambda N_s N$ and $|\mathcal{D}| = (1 - \lambda) N_s N$.

Under coherence disparity, the effective number of pilot symbols usable by each device depends on its individual coherence time and coherence bandwidth. For device $k$ with a coherence block of size $L_{t,k} \times L_{f,k}$, the number of pilot symbols within a single coherence block is
\begin{align}
N_{p,k} \triangleq \Big\lfloor \lambda_t L_{t,k} \Big\rfloor \Big\lfloor \lambda_f L_{f,k} \Big\rfloor.
\label{eq:Npk}
\end{align}
Devices with large coherence bandwidth (i.e., frequency-flat fading) may only require pilot interpolation in time, while devices with frequency-selective fading require pilots spread across both dimensions. Hence, although the pilot lattice is uniformly deployed, its effective utilization is device-specific.

Let $\rho_p$ be the power per pilot symbol and $\rho_d$ the power per data symbol. The total power constraint $\rho$ in \eqref{eq:PSpower-constraint} is satisfied as 
\begin{align}
\rho_p N_p + \rho_d (N_s N - N_p) = \rho N_s N,
\label{eq:PSpower-constraint-OFDMblock}
\end{align}
where $N_p = |\mathcal{P}|$.

\begin{remark}
The general pilot placement strategy above assumes that pilots are needed in both time and frequency. However, in the presence of coherence disparity, this may not be necessary for all devices. For frequency-selective channels, the channel varies across subcarriers, so pilots must be placed across frequency to capture this variation. Here, $L_{f,k} > 1$ and pilot density $\lambda_f > 0$ is necessary. For frequency-flat channels, the channel is constant across all subcarriers (i.e., effectively $L_{f,k} = N$), so only pilots across time are needed. In such cases, allocating multiple pilots in frequency would be wasteful. Thus, the effective number of pilots $N_{p,k}$ should ideally reflect the coherence properties of each device's channel.
\end{remark}

\color{black}
\begin{remark}
The overall pilot fraction in \eqref{eq:pilot-fraction} represents the \emph{network-wide} pilot density, which is determined according to the worst-case links, i.e., devices with the smallest coherence block that require the largest number of pilots to maintain reliable CSI. We discuss this and our device ordering structure (according to coherence block sizes) later in Section~\ref{sec:downlink-signaling}. It is important to note that, under unequal coherence block conditions, the received number of pilots in \eqref{eq:Npk} may exceed the per-device requirement for CSI acquisition. In this work, we exploit this excess pilot overhead through efficient pilot reuse for data transmission, which enables significant communication savings in the FL system.   
\end{remark}
\color{black}

\subsection{Baseline Transmission Scheme}
The baseline transmission scheme uses conventional downlink signaling with \emph{orthogonal} pilots and data symbols, without pilot--payload superposition. That is, pilot and data transmissions are fully separated in both time and frequency. During each FL communication round, indexed by $t$, the PS places known pilot vectors on the pilot positions $\mathcal{P}$ to facilitate channel estimation, and transmits model parameters over the data positions $\mathcal{D}$.

Formally, the transmit signal $\bx[n,m;t] \in \mathbb{C}^M$, across $M$ antennas, at time index $n$ and subcarrier index $m$ during round~$t$ is given by
\begin{align}
    \bx[n,m;t] =
    \begin{cases}
        \sqrt{\rho_p}\,\bx_p[n,m], & (n,m) \in \mathcal{P} \\
        \sqrt{\rho_d}\,\bx_d[n,m;t], & (n,m) \in \mathcal{D}
    \end{cases}
    \label{eq:baseline-transmission}
\end{align}
where $\bx_p[n,m]$ is a known unit-energy pilot reused over time and subcarriers according to the separable pilot lattice (independent of $t$), and $\bx_d[n,m;t]$ carries the model parameters using standard linear modulation, changing across FL rounds.

The average transmit power $\rho$ is satisfied as
\begin{align}
\rho_p\,|\mathcal{P}| + \rho_d\,|\mathcal{D}| \le \rho\,N_s N.
\label{eq:power-constraint-baseline}
\end{align}

\section{Coherence-Aware Device Scheduling and Communication Protocol}
\label{sec:communication-protocol}
In this section, we present the proposed coherence-aware scheduling strategy and a downlink communication protocol designed for wideband fading channels with heterogeneous coherence properties. As described in Section~\ref{sec:system-model}, each device $k$ experiences a coherence block of $L_{t,k}\times L_{f,k}$ channel uses, which determines the required pilot density in time and frequency over the full communication block. Devices with larger coherence blocks require fewer pilots, while devices with smaller coherence blocks require more frequent pilot placement. This coherence disparity, common in practical wireless systems, leads to unequal training overhead across devices and motivates an efficient scheduling and signaling design tailored to their coherence characteristics, which are essential for the practical implementation of FL over wireless.

\subsection{Device Scheduling}
During each communication round $t$, $K$ devices are scheduled to participate in FL; we denote this scheduled set by $\mathcal{K}_t \subseteq \mathcal{K_S} \cup \mathcal{K_D}$ with $|\mathcal{K}_t|=K$. Recall that the device population is partitioned into two groups: (i) static devices $\mathcal{K_S}$, which have large coherence time and coherence bandwidth and therefore stable channels (known CSI), and 
(ii) dynamic devices $\mathcal{K_D}$, which have smaller coherence blocks and require more frequent pilots for accurate channel estimation. 

\color{black}
Let $K_S = |\mathcal{K_S}|$ and $K_D \leq |\mathcal{K_D}|$ respectively denote the number of static and dynamic devices scheduled per round, so that $K_S + K_D = K$. In each round, all $K_S$ static devices are always scheduled, reflecting their stable link conditions and negligible additional pilot overhead. The remaining $K_D=K-K_S$ devices are allocated to dynamic devices, which are selected based on their coherence block sizes. Specifically, the scheduler orders the devices in $\mathcal{K_D}$ according to the product $L_{t,k}L_{f,k}$ and admits the $K_D$ devices with the largest coherence blocks (see Section~\ref{sec:downlink-signaling}). Under the proposed scheme, the number of model parameters that can be reliably delivered to and decoded by device~$k \in \mathcal{K_D}$ in a communication round is proportional to the time--frequency region over which the device has valid CSI, which is determined by its coherence block.

As a result, selecting dynamic devices with larger coherence blocks maximizes the effective number of model parameters
that can be delivered per round under a fixed downlink resource budget. For a scheduled dynamic device~$k$, the coherence block size determines the subset of parameters that can be decoded in that round, captured by the received index set $\mathcal{I}_k(t)$ (defined in Eq.~\eqref{eq:Ik}). The remaining parameters are locally filled using the PLMF strategy described in Section~\ref{sec:plmf-strategy}. By prioritizing devices with larger coherence blocks, the scheduler increases $|\mathcal{I}_k(t)|$, reduces the fraction of missing parameters, and shortens the refresh interval of previously missing coordinates, thereby limiting PLMF-induced drift and stabilizing both local and global model updates, which will be discussed in detail in subsequent sections.
\color{black}
% To select these devices, the scheduler orders $\mathcal{K_D}$ by their coherence block sizes $(L_{t,k}, L_{f,k})$ and admits $K_D$ devices with the largest coherence blocks. This prioritizes dynamic devices whose channels allow more efficient use of downlink resources, while still enabling participation from devices with less favorable coherence conditions across different rounds. 
\begin{remark}
The coherence parameters $(L_{t,k}, L_{f,k})$ are assumed to be available at the PS through lightweight uplink probing or long-term channel statistics, which can be obtained with negligible overhead (details omitted for brevity). The scheduled set $\mathcal{K}_t$ of $K$ devices is then simultaneously served in the downlink protocol and subsequently participates in the uplink aggregation phase.
\end{remark}

\subsection{Downlink Signaling: Integrated Pilot--Parameter Broadcast}
\label{sec:downlink-signaling}
To accommodate scheduling under heterogeneous coherence blocks, the downlink transmission at each FL round serves a dual purpose: (i) providing pilots to dynamic devices for channel estimation, and (ii) broadcasting model updates to all scheduled devices. To realize both efficiently on the same OFDM grid while avoiding excessive pilot overhead, we design a product superposition-based scheme where a portion of the model symbols is embedded within the pilot slots required by dynamic devices, with the remaining symbols carried on data slots. This integrated scheme ensures robust model delivery across diverse coherence conditions.

\color{black}
Without loss of generality, the PS orders the dynamic devices in descending order of their coherence block sizes. Accordingly, for dynamic device indices $i \in \{1,\dots,|\mathcal{K_D}|\}$,
\begin{align*}
    L_{t,i}L_{f,i} > L_{t,i+1}L_{f,i+1}, \quad i = 1, \cdots, |\mathcal{K_D}|-1.
\end{align*}
The scheduler selects the first $K_D$ devices with the largest coherence blocks. Consequently, device indexed $i=K_D$ has the smallest coherence block among the scheduled dynamic devices and therefore requires the largest pilot overhead. This device determines the pilot density required for the downlink signaling design, which corresponds to the network-wide pilot overhead $\lambda$ defined in Section~\ref{sec:pilot-assignment}.
\color{black}

Assume that $s$ symbols are required to broadcast the full global model $\btheta(t)$ in each round.\footnote{The value of $s$ depends on the transmission mode (analog/digital), modulation scheme, coding rate, and quantization level, which together determine its relationship with the model dimension $d$. A detailed treatment of these factors is beyond the scope of this work; see~\cite{Amiri2020FLFading,Amiri2021NoisyDL} for discussion.} Within the downlink super-block of $s$ channel uses\footnote{Set of subcarriers and time slots within the OFDM block.}, the number of pilot intervals is dictated by device~$K_D$, while the specific placement of pilots in time and frequency is governed by the devices with the shortest coherence time and coherence bandwidth, respectively. Define
\begin{align*}
L_{t,s}& \triangleq \min(L_{t,1},\cdots,L_{t,K_D}) \\
L_{f,s}& \triangleq \min(L_{f,1},\cdots,L_{f,K_D}).
\end{align*}
Over each subcarrier $m$, the sequence of $L_{t,s}$ symbols is transmitted across the $M$ antennas. To deliver all $s$ parameters, we assume that $q$ length-$L_{t,s}$ sub-blocks fit within the super-block, each aligned to the coherence time of device~$K_D$. \textcolor{black}{Each length $L_{t,s}$ interval is partitioned into two phases: a \emph{pilot phase} with length~$M$ (because the channel vector on each subcarrier contains $M$ unknown coefficients that must be estimated at dynamic receivers), and a \emph{data phase} with length $L_{t,s}-M$. During communication round~$t$, in sub-block $q'\in\{1,\ldots,q\}$, the transmitted matrix at subcarrier $m$ is
\begin{align}
    &\bX_{q'}(t,m) 
    = \nonumber \\
    &\Big[
\underbrace{\sqrt{\rho_p}\,\,\mathbf{X}^{\theta}_{p,q'}(t,m)\,\mathbf{X}_p}_{M~\text{pilot slots}}
\;\;
\underbrace{\sqrt{\rho_d}\,\,\mathbf{X}^{\theta}_{p,q'}(t,m)\,\mathbf{X}^{\theta}_{d,q'}(t,m)}_{(L_{t,s}-M)~\text{data slots}}
\Big],
    \label{eq:ps-transmit-OFDM}
\end{align}
where $\bX_p\in\mathbb{C}^{M\times M}$ is a \emph{unitary} pilot matrix fixed across sub-blocks, $\bX_{p,q'}^\theta(t,m)\in\mathbb{C}^{M\times M}$ carries $M$ model symbols embedded on pilot positions (via product superposition), and $\bX_{d,q'}^\theta(t,m)\in\mathbb{C}^{M\times(L_{t,s}-M)}$ carries the subsequent $L_{t,s}-M$ model symbols on the data slots\footnote{\textcolor{black}{The partial model matrix $\mathbf{X}^{\theta}_{p,q'}(t,m)$ is designed to be full-rank and invertible. While multiple constructions are possible, one convenient choice is a diagonal embedding,
$\mathbf{X}^{\theta}_{p,q'}(t,m)=\mathrm{diag}(\mathbf{s}_{p,q'}(t,m))$,
where $\mathbf{s}_{p,q'}(t,m)\in\mathbb{C}^M$ has nonzero entries. The subsequent $L_{t,s}-M$ model symbols are mapped column-wise into $\mathbf{X}^{\theta}_{d,q'}(t,m)$ and transmitted across the $M$ antennas during the data phase.}}.}

Transmitting~\eqref{eq:ps-transmit-OFDM}, the received signal at device~$k$ on subcarrier~$m$ in sub-block~$q'$ is
\begin{align}
    \by_{k,q'}[m] &= \bh_k[m]^H \nonumber\\
    \times & \Big(\sqrt{\rho_p}\,\bX_{p,q'}^\theta(t,m)\,\bX_p \;\;\; \sqrt{\rho_d}\,\bX_{p,q'}^\theta(t,m)\,\bX_{d,q'}^\theta(t,m)\Big) \nonumber\\
    +&\, \bw_{k,q'}[m].
    \label{eq:receivedsignal-ps}
\end{align}
Any static device $k\in\mathcal{K_S}$ can directly decode both pilot-carried and data-carried parameters since it has reliable CSI and knowledge of $\bX_p$. The same holds for dynamic devices whose channels remain constant within the block. 

In contrast, any scheduled device with time- and/or frequency-varying channels must first perform channel estimation using the pilot slots. However, such devices cannot directly recover the exact link gain $\bh_k[m]$, since the embedded parameter matrix $\bX_{p,q'}^\theta(t,m)$ is unknown. Instead, they estimate the \emph{equivalent channel}, defined as the product of the link gain and the embedded parameter matrix over pilot slots
\begin{align}
    \mathbf{f}_{k,q'}[m] \triangleq \bh_k[m]^H \bX_{p,q'}^\theta(t,m),
    \label{eq:ps-equivalent}
\end{align}
during the pilot phase. This estimate is then used to coherently decode the data-carried symbols in $\bX_{d,q'}^\theta(t,m)$.

The MMSE estimate of $\mathbf{f}_{k,q'}[m]$ is \cite{Hassibi2003HowMuch}
\begin{align}
   \overline{\mathbf{f}}_{k,q'}&[m] \nonumber \\
   =&\, \mathbb{E}\big[\mathbf{f}_{k,q'}[m]\,\by_{k,q'}[m]^H\big]\,
   \mathbb{E}\big[\by_{k,q'}[m]\,\by_{k,q'}[m]^H\big]^{-1}\by_{k,q'}[m] \nonumber\\
   =&\, \frac{M\rho_p}{M\rho_p+\sigma_w^2}\Big(\mathbf{f}_{k,q'}[m] + \bw_{k,q'}[m]\Big).
   \label{eq:MMSE-equivalent-OFDM}
\end{align}
The estimation error $\tilde{\mathbf{f}}_{k,q'}[m] = \mathbf{f}_{k,q'}[m] - \overline{\mathbf{f}}_{k,q'}[m]$ is Gaussian with covariance $\sigma_{e,k}^2 \bI$, such that \cite{Hassibi2003HowMuch}
\begin{align}
    \sigma_{e,k}^2 = \frac{M\sigma_w^2}{M\rho_p+\sigma_w^2}.
    \label{eq:estimation-error-OFDM}
\end{align}
\begin{remark}
The proposed signaling efficiently accommodates heterogeneous coherence conditions by embedding parameters within pilot slots through product superposition. This approach enables full model delivery to static devices—whose stable channels allow immediate decoding—while simultaneously supporting dynamic devices, which can coherently recover partial updates after channel estimation. By jointly serving both device types, the scheme avoids wasting downlink resources that would otherwise be consumed if static and dynamic devices were treated separately\footnote{Known as orthogonal signaling.}. This co-scheduling capability reduces communication latency in federated learning over practical wireless channels, which commonly exhibit coherence disparity across devices.
\end{remark}
\color{black}
\begin{remark}
\label{re:coherence-misalignment}
We assume that $s$ transmit symbols are used to deliver $d$ model parameters. The proposed signaling is stated in a general form and is compatible with various transmission modes, modulation and coding schemes, and quantization levels that determine the relationship between $s$ and $d$. Importantly, the value of $s$ does not change the structure of the proposed pilot--parameter superposition scheme; it only affects the number of required sub-blocks~$q$, while the proposed signaling within each sub-block remains effective and unchanged.
\vspace{-0.5cm}
% While we assume that $s$ is an integer multiple of $L_{t,s}$ and that downlink signaling aligns with the shortest coherence block for simplicity and analytical tractability, the scheme extends naturally to non-integer cases and misalignment. The PS can apply the same pilot--parameter superposition principles across any block size and flexibly shift pilots in time and frequency to match heterogeneous coherence conditions. Multiple superpositions may also be applied within a block without interference, enabling efficient joint service of static and dynamic devices.
\end{remark}
\color{black}

\subsection{Downlink Impact on the Local Model}
\label{sec:impact-on-learning}
In the following, we characterize how the downlink signaling affects the model available at each scheduled device. Because pilot placement and parameter mapping are coherence-aware, the set of parameters delivered to a device is \emph{deterministic}; randomness arises only from fading, receiver noise, and channel-estimation error.

\subsubsection{Deterministic coverage and received indices}
At round $t$, the PS fixes a deterministic placement of the $s$ downlink symbols--which jointly encode the global model $\btheta(t)$--across the OFDM grid $\mathcal{G}$ (as per the signaling in Section~\ref{sec:downlink-signaling}). 
For device $k$, let $\mathcal{R}_k(t)\subseteq\mathcal{G}$ be the set of time-frequency positions that fall within its coherence block(s) and include pilots enabling coherent equalization. 
Let $\pi:\{1,\ldots,s\}\to\mathcal{G}$ denote the placement map from symbol indices to grid positions. 
The set of symbols successfully delivered to device~$k$ is then
\begin{align}
\mathcal{I}_k(t) 
\triangleq \big\{\, i\in\{1,\ldots,s\}:\ \pi(i)\in\mathcal{R}_k(t)\,\big\}.
\label{eq:Ik}
\end{align}
Thus, $\mathcal{I}_k(t)$ is coherence-aware and determined by the pilot-data allocation; the complement $\overline{\mathcal{I}}_k(t)\triangleq\{1,\ldots,s\}\setminus\mathcal{I}_k(t)$ corresponds to symbols not delivered in round~$t$. 
After reception, device~$k$ reconstructs its local version of $\btheta(t)$ by filling missing entries in $\overline{\mathcal{I}}_k(t)$ using the PLMF (previous local model filling) strategy, as discussed later in this section.

\subsubsection{Per-entry distortion on received coordinates}
We now characterize the distortion on the coordinates successfully received by a dynamic device~$k$. Recall from the downlink signaling protocol (Section~\ref{sec:downlink-signaling}) that such a device estimates an \emph{equivalent channel}, $\bf_{k,q'}[m]$, which is the product of the true channel and the parameters intended for static devices. The received signal for the data portion, as referenced in \eqref{eq:receivedsignal-ps}, can be decomposed based on the channel estimate $\overline{\bf}_{k,q'}[m]$ and its error $\tilde{\bf}_{k,q'}[m]$
\begin{align}
    \by_{k,q'}(t,m) =&\,\sqrt{\rho_d}\,\overline{\bf}_{k,q'}(t,m)\bX_{d,q'}^\theta(t,m) \nonumber \\
    +&\, \sqrt{\rho_d}\,\tilde{\bf}_{k,q'}(t,m)\bX_{d,q'}^\theta(t,m) \nonumber \\
    +&\, \bw_{k,q'}(t,m).
    \label{eq:yk-dl}
\end{align}
The first term in \eqref{eq:yk-dl} is the desired signal, while the second and third terms represent the distortion from residual channel estimation error and receiver AWGN, respectively. The effective SNR for device~$k$ is therefore
\begin{align}
\gamma_k[n,m] = \frac{\rho_d\,\mathbb{E}\big[|\overline{\bf}_k[n,m]|^2\big]}{\rho_d\,\sigma_{e,k}^2+\sigma_w^2},
\label{eq:snr-dl}
\end{align}
where the numerator is the average power of the desired signal component, and the denominator is the sum of channel estimation error and noise powers. Let $c_k \triangleq \frac{\rho_d}{\rho_d\,\sigma_{e,k}^2+\sigma_w^2}$. For unit-variance Gaussian symbols, the linear MMSE distortion for a given channel realization is
\begin{align}
v_k[n,m] = \frac{1}{1+c_k\,|\overline{\bf}_k[n,m]|^2}.
\label{eq:vk-nm}
\end{align}
Averaging over the channel distribution yields the expected per-entry distortion on received coordinates
\begin{align}
v_k \triangleq \mathbb{E}\!\left[\frac{1}{1+c_k\,|\overline{\bf}_k|^2}\right].
\label{eq:vk-avg}
\end{align}

\subsubsection{PLMF for missing parameters}
\label{sec:plmf-strategy}
Define the diagonal selection matrix $\bD_k(t)\in\{0,1\}^{d\times d}$ with $[\bD_k(t)]_{ii}=1$ if $i\in\mathcal{I}_k(t)$ and $0$ otherwise. 
Let
\begin{align}
\zeta_{k,i}(t) \triangleq \max\big\{t' \le t:\ i\in\mathcal{I}_k(t')\big\}
\label{eq:last-recv-index}
\end{align}
be the most recent round (no later than $t$) when device $k$ received symbol/coordinate $i$; if $i\in\mathcal{I}_k(t)$ then $\zeta_{k,i}(t)=t$. Thus, $\zeta_{k,i}(t)$ precisely indexes the last round in which device $k$ successfully refreshed coordinate $i$, and will serve as the per-coordinate recency marker in what follows.

The PLMF rule is
\begin{align}
\widehat{\theta}_{k,i}^{(t)} =
\begin{cases}
\theta_i^{(t)} + \varepsilon_{k,i}^{(t)}, & i\in\mathcal{I}_k(t),\\[0.5ex]
\widehat{\theta}_{k,i}^{\big(\bzeta_{k,i}(t)\big)}, & i\notin\mathcal{I}_k(t),
\end{cases}
\label{eq:plmf-rule}
\end{align}
where $\varepsilon_{k,i}^{(t)}$ is zero-mean with variance $v_k$ on received entries. 
Stacking across coordinates gives the vector form
\begin{align}
\widehat{\btheta}_k^{(t)}
= \bD_k(t)\big(\btheta^{(t)}+\bepsilon_k^{(t)}\big)
+ \big(\bI-\bD_k(t)\big)\,\widehat{\btheta}_k^{\big(\bzeta_k(t)\big)},
\label{eq:theta-plmf}
\end{align}
where $\bepsilon_k^{(t)}$ stacks $\{\varepsilon_{k,i}^{(t)}\}_{i=1}^d$ and $\bzeta_k(t)$ stacks $\{\zeta_{k,i}(t)\}_{i=1}^d$. 
The covariance of $\bepsilon_k^{(t)}$ is
\[
\mathbb{E}\!\left[\bepsilon_k^{(t)} \big(\bepsilon_k^{(t)}\big)^{\!H}\right]  = v_k\,\bD_k(t).
\]
% Here, $[\widehat{\btheta}_k^{(\bzeta_k(t))}]_i = \widehat{\theta}_{k,i}^{(\bzeta_{k,i}(t))}$; that is, the superscript denotes per-coordinate indexing by the most recent reception round and is not to be interpreted as exponentiation.
Let $[\widehat{\btheta}_k^{(\bzeta_k(t))}]_i \triangleq \widehat{\theta}_{k,i}^{(\bzeta_{k,i}(t))}$. The PLMF drift on a missing coordinate is
\begin{align}
r_{k,i}^{(t)} &\triangleq \theta_i^{(t)} - \widehat{\theta}_{k,i}^{\big(\bzeta_{k,i}(t)\big)} \nonumber \\
&= \sum_{t'=\bzeta_{k,i}(t)+1}^{t}\big(\theta_i^{(t')}-\theta_i^{(t'-1)}\big),
\label{eq:drift}
\end{align}
which will be controlled under standard smoothness and stepsize conditions in the convergence section.

\begin{remark}
In this design, missing parameters at each device are \emph{not} random: $\mathcal{I}_k(t)$ is a designed, coherence-aware subset. The principal randomness is the physical-layer distortion on received entries, summarized by $v_k$ in \eqref{eq:vk-avg}. 
This separation is used explicitly in the convergence proofs.
\end{remark}

\subsection{Limits for the Downlink}
\label{sec:it-downlink-dl}
We now quantify the downlink performance in terms of achievable rates for static and dynamic devices under the product-superposition signaling in \eqref{eq:ps-transmit-OFDM}, and optimize the pilot/data powers under the round-level constraint.

\subsubsection{Power budget over the OFDM super-block}
The $s$ channel uses per round are partitioned into $q$ sub-blocks of length $L_{t,s}$, each containing $M$ pilot slots and $L_{t,s}-M$ data slots. 
Let $\rho_p$ and $\rho_d$ denote the per-symbol pilot and data powers. 
The average power constraint over the round is
\begin{align}
    qM\big(\rho_p + \rho_d(L_{t,s}-M)\big) \leq \rho s.
    \label{eq:power-constraint-OFDM}
\end{align}

\subsubsection{Achievable rate for static devices}
Pilot slots are reused to carry model symbols to static devices. 
For a static device $k'\in\mathcal{K_S}$, in sub-block~$q^\prime$, the achievable rate per channel use is
\begin{align}
    R_{k',q'} 
    =&\, \frac{M}{L_{t,s}} \,\mathbb{E}\!\left[\log_2\!\left(1 + \frac{\rho_p}{M\sigma_w^2}\,\|\bh_{k'}\|^2\right)\right] \nonumber \\
    &+\frac{L_{t,s}-M}{L_{t,s}} \,\mathbb{E}\!\left[\log_2\!\big(1+\frac{\rho_d}{M\sigma_w^2}\,\|\bh_{k'}\|^2\big)\right]
    \label{eq:rate_static}
\end{align}
where the first term denotes the achievable rate over the pilot slots and the second term denotes the achievable rate over the data slots within the sub-block.

\subsubsection{Achievable rate for dynamic devices}
Dynamic devices decode only during data slots. 
Using the MMSE equivalent-channel estimate from the pilot phase, their effective SNR in sub-block $q'$ is
\begin{align}
    \gamma_{k,q'} = 
    \frac{\rho_d}{\rho_d\sigma_{e,k}^2+\sigma_w^2}\,\|\overline{\bf}_{k,q'}\|^2,
    \label{eq:snr-dynamic}
\end{align}
where $\overline{\bf}_{k,q'}$ is the MMSE estimate of the equivalent channel (see Eqs.~\eqref{eq:MMSE-equivalent-OFDM}--\eqref{eq:estimation-error-OFDM}). 
The corresponding achievable rate is
\begin{align}
    R_{k,q'} 
    = \frac{L_{t,s}-M}{L_{t,s}} \,\mathbb{E}\!\left[\log_2\!\big(1+\gamma_{k,q'}\big)\right].
    \label{eq:rate_dynamic}
\end{align}

\subsubsection{Pilot–data power allocation}
Dynamic links are the bottleneck; hence we select $(\rho_p,\rho_d)$ to maximize \eqref{eq:rate_dynamic} subject to \eqref{eq:power-constraint-OFDM}. 
This yields the closed-form optimum
\begin{align}
    \rho_d^* &= \frac{\sigma_w^2+\rho L_{t,s}}{M\sqrt{L_{t,s}-M}\,\big(1 + \sqrt{L_{t,s}-M}\big)},
    \label{eq:rho_d_opt} \\[0.5ex]
    \rho_p^* &= \frac{\rho L_{t,s}}{M} - \rho_d^*(L_{t,s}-M),
    \label{eq:rho_p_opt}
\end{align}
whose derivation follows the same steps as in the single-carrier case and is deferred to Appendix~\ref{app:proof-PowerAllocation}.

\begin{remark}
The allocation in \eqref{eq:rho_d_opt}--\eqref{eq:rho_p_opt} balances pilot quality (reducing $\sigma_{e,k}^2$) and data power (boosting $\gamma_{k,q'}$). 
Combined with pilot-slot reuse for statics in \eqref{eq:rate_static}, the scheme avoids orthogonalization across device types and reduces training latency under heterogeneous coherence conditions.
\end{remark}

\section{Uplink OTA Aggregation and Global Update}
\label{sec:uplink-signal-model}
We now present the uplink transmission and OTA aggregation protocol at the PS. The design reuses the OFDM grid and scheduling policy from Section~\ref{sec:communication-protocol}. Recall that each communication round~$t$ consists of: (i) downlink transmissions (pilots and parameters) over an OFDM super-block; (ii) $\tau$ local SGD steps at each scheduled device; (iii) an uplink OFDM super-block used for OTA aggregation over the MAC channel.

We consider an \emph{analog} uplink transmission over the MAC channel for OTA aggregation at the PS. Recall that after $\tau$ SGD steps, device~$k$ sends $\Delta \btheta_k^{(t)} = \btheta_{k,\tau}^{(t)} - \btheta_{k,0}^{(t)}$ to the PS (see Section~\ref{sec:system-model}). Let $a_k \triangleq \frac{B_k}{\sum_{\ell\in\mathcal{K}_t} B_\ell}$, and write $\Delta \btheta_k^{(t)}=\big[\Delta \theta_{k,1}^{(t)},\ldots,\Delta \theta_{k,d}^{(t)}\big]^T$. For parameter $i\in[d]$, the aggregation model is
\begin{align}
S_i^{(t)} = \sum_{k\in\mathcal{K}_t} a_k\, \Delta \theta_{k,i}^{(t)}.
\label{eq:ul-target-sum}
\end{align}

\subsection{Coherence-Aware Uplink Protocol}
A key challenge in the uplink is that scheduled devices $k \in \mathcal{K}_t$ have heterogeneous and unaligned coherence blocks. A naive pilot transmission scheme where each device places a pilot within its own block would lead to pilot collisions and a complex, unsynchronized reception process at the PS.

To overcome this, we introduce a structured uplink protocol that leverages the coherence-aware scheduling policy from Section~\ref{sec:communication-protocol}. The scheduler already prioritizes dynamic devices with the largest coherence blocks. This ensures the devices in $\mathcal{K}_t$ are relatively homogeneous in their channel dynamics, which makes a synchronized protocol efficient.

The protocol is structured as follows:
\begin{enumerate}
    \item \emph{Sub-block partitioning:} As defined in Section~\ref{sec:downlink-signaling}, $L_{t,s}$ and $L_{f,s}$ respectively denote the shortest temporal and spectral coherence lengths among the scheduled devices. The uplink super-block (OFDM block) is partitioned into uniform rectangular sub-blocks of size $L_{t,s}\times L_{f,s}$, indexed by $p\in\{1,\ldots,P_{\mathrm{ul}}\}$. The PS identifies this size as the bottleneck coherence block, where the channels remain constant for each scheduled device within a sub-block.
    
    \item \emph{Two-phase transmission:} Within each sub-block, communication is organized into two phases: a \emph{channel training phase} for pilot transmission, followed by an \emph{OTA aggregation phase} for data (parameter updates) transmission.\footnote{This partitioning can be implemented in either the time or frequency domain, since the channel within each sub-block is time-invariant and frequency-flat. In this work, we assume that the channel is estimated on a single subcarrier within the sub-block, while the remaining subcarriers are used for OTA data transmission over the MAC channel.}
\end{enumerate}
This design creates a predictable grid for both channel estimation and data aggregation, ensuring robust operation for all scheduled devices.

\subsection{Channel Estimation in Uplink Sub-Blocks}
\label{subsec:ul-ce}
At the beginning of each uplink sub-block $p$, only the \emph{dynamic} devices $\mathcal{K_D}$ transmit one orthogonal pilot symbol (e.g., disjoint tones or orthogonal codes) with power $\rho_\tau$. Let $\tilde{\bh}_{k,p}\in\mathbb{C}^{M}$ be the uplink channel from device~$k$ to the PS. The PS obtains the MMSE estimate $\widehat{\bh}_{k,p}$ with error variance
\begin{align}
\mathbb{E}\!\left[\|\tilde{\bh}_{k,p}-\widehat{\bh}_{k,p}\|_2^2\right] = M\,\sigma_{e,k}^2,
\label{eq:ul-mmse-error}
\end{align}
with $\sigma_{e,k}^2 = \frac{\sigma_w^2}{\rho_\tau+\sigma_w^2}$~\cite{Hassibi2003HowMuch}. Static devices $\mathcal{K_S}$ do not send uplink pilots in the round; the PS retains their channel state from previous estimates.  
The PS applies a single unit-norm receive combiner $\bu_p\in\mathbb{C}^{M}$ for the entire sub-block.

% \subsection{Coordinate Assignment and Masked Increments}
% \label{subsec:ul-assignment}

\subsection{MAC Transmission and PS Observation}
\label{subsec:ul-mac}
Let $\{\mathcal{I}_p\}_{p=1}^{P_{\mathrm{ul}}}$ be a disjoint partition of $\{1,\ldots,d\}$; sub-block $p$ carries coordinates in $\mathcal{I}_p$.  
After $\tau$ local steps, device~$k$ forms its increment $\Delta\btheta_k^{(t)}\in\mathbb{R}^d$. The deterministic downlink mask $\bD_k(t)\in\{0,1\}^{d\times d}$ enforces
\begin{align}
\widetilde{\Delta\btheta}_k^{(t)} \triangleq \bD_k(t)\,\Delta\btheta_k^{(t)}.
\label{eq:masked-increment}
\end{align}
For $i\in\mathcal{I}_p$, only $[\widetilde{\Delta\btheta}_k^{(t)}]_i$ is transmitted in sub-block $p$. 

Define the effective scalar channel $g_{k,p}\triangleq \bu_p^H\tilde{\bh}_{k,p}$ (i.e., the post-combining channel within sub-block $p$, using the MMSE estimates from Section~\ref{subsec:ul-ce}). Each scheduled device transmits its assigned coordinates with per-symbol power $\rho_u$ and a common scaling $\beta_p>0$
\begin{align}
x_{k,p}^i = \sqrt{\rho_u}\,\alpha_{k,p} a_k\,\big[\widetilde{\Delta\btheta}_k^{(t)}\big]_i,
\label{eq:ul-symbol}
\end{align}
where $\alpha_{k,p} \triangleq \frac{\beta_p}{\max\{|g_{k,p}|\,,\,\mu_k\}}\,e^{-j\angle g_{k,p}}, \mu_k>0$\footnote{The computation of the precoder $\alpha_{k,p}$ requires each device to know its effective scalar channel $g_{k,p}$. This value is computed at the PS after the channel training phase and is then broadcast back to the scheduled devices. This feedback is assumed to occur over a low-latency downlink control channel, and its overhead is considered negligible compared to the model parameter transmission.}, where $\mu_k$ is a per-device clipping floor that prevents unbounded amplification under deep fades. Here, $\beta_p>0$ is a deterministic normalization constant chosen by the PS (per sub-block $p$) to stabilize the received OTA magnitude across sub-blocks given the participating devices and their channels. This yields
\begin{align}
\frac{1}{\sum_{p=1}^{P_{\mathrm{ul}}}|\mathcal{I}_p|}\sum_{p=1}^{P_{\mathrm{ul}}}\sum_{i\in\mathcal{I}_p}\mathbb{E}\!\big[|x_{k,p}^i|^2\big] \le P_k,
\label{eq:ul-power}
\end{align}
where $P_k$ is the power budget for device~$k$. After receive combining with $\bu_p$, the PS obtains, for each $i\in\mathcal{I}_p$,
\begin{align}
r_{i,p}
&= \bu_p^H \sum_{k\in\mathcal{K}_t}\tilde{\bh}_{k,p}x_{k,p}^i + z_{i,p} \nonumber\\
&= \sqrt{\rho_u}\,\beta_p\sum_{k\in\mathcal{K}_t}\chi_{k,p} a_k\,\big[\widetilde{\Delta\btheta}_k^{(t)}\big]_i + z_{i,p},
\label{eq:ps-scalar}
\end{align}
with $\chi_{k,p} \triangleq \frac{|g_{k,p}|}{\max\{|g_{k,p}|\,,\,\mu_k\}} \in (0,1]$ and $z_{i,p}\sim \mathcal{CN}\!\big(0,\,\sigma_w^2\|\bu_p\|_2^2\big) = \mathcal{CN}(0,\sigma_w^2)
$.

\subsection{Per-Coordinate Estimator and Global Update}
\label{subsec:ul-estimate}
For $i\in\mathcal{I}_p$, the PS forms
\begin{align}
\widehat{\Delta\theta}_i^{(t)} 
\triangleq \frac{r_{i,p}}{\sqrt{\rho_u}\,\beta_p}
= \sum_{k\in\mathcal{K}_t}\chi_{k,p} a_k\,\big[\bD_k(t)\Delta\btheta_k^{(t)}\big]_i + \tilde{z}_{i,p},
\label{eq:ul-estimator}
\end{align}
with $\tilde{z}_{i,p}\triangleq \frac{z_{i,p}}{\sqrt{\rho_u}\,\beta_p}$. Note that $\chi_{k,p}$ introduces a per-device multiplicative bias relative to the ideal sum in \eqref{eq:ul-target-sum}; this bias is carried into the error decomposition and convergence analysis in Section~\ref{sec:convergence}. Stacking over all coordinates yields
\begin{align}
\btheta^{(t+1)} 
= \btheta^{(t)} + \widehat{\Delta\btheta}^{(t)},
\label{eq:ul-update}
\end{align}
with 
\begin{align}
    \widehat{\Delta\btheta}^{(t)} = \sum_{k\in\mathcal{K}_t} \overline{\boldsymbol{\chi}} a_k\,\bD_k(t)\,\Delta\btheta_k^{(t)} + \bn_{\mathrm{ul}}^{(t)},
\end{align}
where $\overline{\boldsymbol{\chi}}$ distributes the sub-block coefficients $\chi_{k,p}$ over their corresponding coordinates $\{\mathcal{I}_p\}$, and $\bn_{\mathrm{ul}}^{(t)}$ stacks the zero-mean Gaussian terms with covariance
\begin{align}
\mathbb{E}\!\left[\bn_{\mathrm{ul}}^{(t)}\big(\bn_{\mathrm{ul}}^{(t)}\big)^H\right] 
&= \frac{\sigma_w^2}{\rho_u}\,\mathbb{E}\!\left[\frac{\|\bu_p\|_2^2}{\beta_p^2}\right]\bI_d \nonumber \\
&= \frac{\sigma_w^2}{\rho_u}\,\mathbb{E}\!\left[\frac{1}{\beta_p^2}\right]\bI_d.
\label{eq:ul-cov}
\end{align}

\begin{remark}
Applying the mask $\bD_k(t)$ before OTA aggregation has two important consequences. First, the device transmits no power (i.e., is silent) for any coordinate $i$ where $[\widetilde{\Delta\btheta}_k^{(t)}]_i = 0$. The OTA process via the MAC channel naturally handles this by summing the signals only from the non-silent devices, correctly computing the aggregate over the subset of devices with valid updates. Second, from a learning perspective, this per-coordinate partial aggregation introduces a manageable bias, which is a known trade-off in practical wireless FL. This design is deliberate, as aggregating valid updates from a subset of devices is fundamentally preferable to corrupting the global model with stale, invalid updates. We formally account for the impact of this bias in our convergence analysis in Section~\ref{sec:convergence}.
\end{remark}

\color{black}
\begin{remark}
The proposed coherence-aware scheduling rule is not claimed to be globally optimal, but is a simple, physically motivated policy aligned with coherence-limited downlink delivery and additional degrees of freedom gain enabled by pilot reuse. Joint optimization of scheduling under coherence disparity, data heterogeneity, and link-quality constraints is beyond the scope of this work and is left for future work.
\end{remark}
\color{black}
\vspace{-.5cm}
% \subsection{Overhead Accounting}
% \label{subsec:ul-overhead}
% In each uplink sub-block, dynamic devices $\mathcal{K_D}$ consume $|\mathcal{K_D}|$ orthogonal pilot symbols, while static devices do not transmit pilots. The remaining symbols in the sub-block carry OTA data for the coordinates in $\mathcal{I}_p$. Because every sub-block follows the same pilot–data cadence, the PS has deterministic switching points between channel estimation and data combining within the uplink super-block.

\section{Convergence Analysis}
\label{sec:convergence}
This section analyzes the convergence behavior of the proposed coherence-aware FL under imperfect downlink and uplink channels: OFDM-based product superposition with PLMF on the downlink and imperfect OTA aggregation over the MAC. We isolate deterministic, coherence-driven coverage from stochastic distortions due to fading, estimation, receiver noise, and local SGD, yielding a clean bias–variance view of the end-to-end update. The resulting bounds quantify how the proposed scheme operates under heterogeneous coherence patterns and how these effects shape optimization in convex, strongly convex, and nonconvex regimes.

\subsection{Assumptions}
Let $\{\mathcal{F}_t\}_{t \ge 0}$ be the natural filtration generated by all randomness up to the start of round $t$. Our analysis relies on the following standard assumptions.

\begin{itemize}
    \item[\emph{(A1)}] \emph{$L$-Smoothness:} Each local loss function $F_k$ is $L$-smooth. Consequently, the global loss function $F$ is also $L$-smooth:
    \begin{align}
        \|\nabla F_k(\btheta_a) - \nabla F_k(\btheta_b)\| \le L \|\btheta_a - \btheta_b\|, \quad \forall \btheta_a, \btheta_b.
    \end{align}

    \item[\emph{(A2)}] \emph{Bounded Gradient Variance:} The stochastic gradients computed on any device $k$ are unbiased and have bounded variance. For any model $\btheta$,
    \begin{align}
        &\mathbb{E}[\nabla f(\btheta, \bbeta_k) \mid \mathcal{F}_t] = \nabla F_k(\btheta), \\
        &\mathbb{E}[\|\nabla f(\btheta, \bbeta_k) - \nabla F_k(\btheta)\|^2 \mid \mathcal{F}_t] \le \sigma_g^2.
    \end{align}

    \item[\emph{(A3)}] \emph{PLMF and Downlink Noise:} For each device $k$, the coherence-aware downlink placement delivers a \emph{deterministic fraction} of parameters per round, denoted $q_k \in [0,1]$, i.e., $q_k \equiv |\mathcal{I}_k(t)|/s$ under the fixed placement (constant or slowly varying across $t$). The decoding distortion on a delivered parameter, $e_{k,i}^{(t)}$, is zero-mean with variance $v_k$. For missing coordinates, PLMF uses the freshest available value; the resulting per-coordinate drift
    \begin{align}
        r_{k,i}^{(t)} \triangleq \theta_i^{(t)} - \widehat{\theta}_{k,i}^{\big(\zeta_{k,i}(t)\big)}
    \end{align}
    has a uniformly bounded second moment
    \begin{align}
        \mathbb{E}\!\left[(r_{k,i}^{(t)})^2 \mid \mathcal{F}_t\right] \le D^2.
    \end{align}
    Here, $D^2$ is a finite constant, independent of $t$, that uniformly bounds the drift variance across all devices and coordinates, where $\zeta_{k,i}(t)\!\le\!t$ is the last round index at which coordinate $i$ of device $k$ was refreshed (see Section~\ref{sec:impact-on-learning}). 

    \item[\emph{(A4)}] \emph{Bounded Error Propagation Through Local SGD:} The local update map is stable under $L$-smoothness and bounded stepsizes. Specifically, letting $u_{k,i}(\cdot)$ denote the $i$th coordinate of the (vector) local update after $\tau$ steps, the deviation induced by an inexact local model satisfies
    \begin{align}
        \big|u_{k,i}(\widehat{\btheta}_k^{(t)}) - u_{k,i}(\btheta^{(t)})\big|
        \le \kappa_\tau\,\big|\delta_{k,i}^{(t)}\big| + \xi_{k,i}^{(t)},
    \end{align}
    where $\delta_{k,i}^{(t)} \triangleq \widehat{\theta}_{k,i}^{(t)} - \theta_i^{(t)}$, $\kappa_\tau$ depends on $L$ and the local stepsizes, and $\xi_{k,i}^{(t)}$ is zero-mean accumulated SGD noise with bounded variance, uniformly in $\tau$.
\end{itemize}

\color{black}
\begin{remark}
Assumption~(A3) captures the effect of reusing previously received model coordinates under partial downlink reception. In the proposed framework, the global model is continuously updated in every round using full model updates from static devices. Moreover, due to coherence disparity, the set of missing coordinates is generally not aligned across dynamic devices; a coordinate missing at one device may be refreshed through other scheduled devices. As a result, the global model does not experience unbounded staleness in any coordinate, and the resulting PLMF-induced drift remains controlled under standard smoothness and bounded stepsize conditions.   
\end{remark}
\color{black}
\vspace{-.5cm}
\subsection{Aggregation Error Decomposition}
For any model coordinate $i$ in round $t$, the aggregated update received at the PS, $\hat{S}_i$, can be decomposed as
\begin{align}
    \hat{S}_i = S_i^* + E_i^{\mathrm{dl}} + n_i^{\mathrm{ul}},
\end{align}
where $S_i^* \triangleq \sum_{k \in \mathcal{K}_t} a_k\, u_{k,i}(\btheta^{(t)})$ is the \emph{ideal} update computed from the true global model, $E_i^{\mathrm{dl}} \triangleq \sum_{k \in \mathcal{K}_t} a_k\big(u_{k,i}(\widehat{\btheta}_k^{(t)}) - u_{k,i}(\btheta^{(t)})\big)$ captures the downlink/PLMF-induced error, and $n_i^{\mathrm{ul}}$ is zero-mean OTA aggregation noise from the uplink. We define the following variance bounds, all assumed finite and uniform across rounds:
\begin{align}
    \sigma_{\mathrm{ul}}^2 &\triangleq \sup_{t,i} \mathbb{E}\!\left[|n_i^{\mathrm{ul}}|^2 \mid \mathcal{F}_t\right], \\
    \sigma_{\mathrm{dl}}^2 &\triangleq \sup_{t,i} \mathbb{E}\!\left[\big|E_i^{\mathrm{dl}} - \mathbb{E}[E_i^{\mathrm{dl}}\mid\mathcal{F}_t]\big|^2 \mid \mathcal{F}_t\right], \\
    \Xi_t &\triangleq \sigma_g^2 + d\big(\sigma_{\mathrm{ul}}^2 + \sigma_{\mathrm{dl}}^2\big),
\end{align}
and let $\overline{\Xi} \triangleq \sup_t \Xi_t$. The multiplicative factor $d$ accounts for stacking the per-coordinate variances across the $d$-dimensional update vector.

The PLMF bias vector is defined in the update space as
\begin{align}
    \boldsymbol{\mathcal{B}}_t 
    \triangleq 
    \mathbb{E}\!\left[
        \sum_{k \in \mathcal{K}_t} a_k \big(u_k(\widehat{\btheta}_k^{(t)}) - u_k(\btheta^{(t)})\big)
        \,\Big|\, \mathcal{F}_t
    \right].
\end{align}
Its bounded squared norm is denoted
\begin{align}
    \overline{B} \triangleq \sup_t \,\mathbb{E}\!\left[\|\boldsymbol{\mathcal{B}}_t\|^2\right].
\end{align}
This captures the systematic bias induced by partial parameter delivery and PLMF drift. The bias vanishes when $q_k=1$ for all devices, and otherwise scales with both the missing ratio $1-q_k$ and the drift magnitudes $\{r_{k,i}^{(t)}\}$. By (A4), this bias is controlled by the local-SGD stability constant $\kappa_\tau$ together with the second moments of the per-coordinate mismatches $\{\delta_{k,i}^{(t)}\}$.

\vspace{-0.1in}
\subsection{Main Results}
%We now state the main convergence guarantees for our proposed framework.

\begin{theorem}
\label{thm:convex}
Assume each local loss function $F_k$ is L-smooth and the global loss function $F$ is convex. Let $\btheta^*$ be a minimizer of $F$, and $F^* = F(\btheta^*)$. Under assumptions (A1)–(A4):
\begin{enumerate}
    \item (Convex Case) If the learning rate is constant $\overline{\eta} \le 1/(4L)$, the average iterate $\overline{\btheta}^{(T)} \triangleq \frac{1}{T}\sum_{t=1}^{T}\btheta^{(t)}$ satisfies
    \begin{align}
        \mathbb{E}\!\left[F(\overline{\btheta}^{(T)}) - F^*\right] \le \frac{\|\btheta^{(1)} - \btheta^*\|^2}{2\overline{\eta}T} + \overline{\eta}\big(L\overline{\Xi} + 2\overline{B}\big).
    \end{align}
    \item ($\mu$-Strongly Convex Case) If $F$ is also $\mu$-strongly convex, and we use a diminishing step size $\overline{\eta}_t = \frac{\beta}{\mu(t+\gamma)}$ with $\beta > 1$ and $\gamma$ large enough such that $\overline{\eta}_1 \le 1/(4L)$, then for $T \ge 1$:
    \begin{align}
        \mathbb{E}\!\left[\|\btheta^{(T+1)} - \btheta^*\|^2\right] \le \frac{\nu}{T+\gamma},
    \end{align}
    where $\nu \triangleq \max\!\left( (\gamma+1)\,\mathbb{E}\!\left[\|\btheta^{(1)}-\btheta^*\|^2\right], \frac{\beta^2 C_{\mathrm{err}}}{\mu^2(\beta-1)} \right)$ and $C_{\mathrm{err}} \triangleq 2L\,\overline{\Xi} + 4\,\overline{B}$ is the combined error constant.
\end{enumerate}
\end{theorem}

\begin{proof}
    See Appendix~\ref{app:proof-convex}.
\end{proof}

\begin{theorem}
\label{thm:nonconvex}
Assume each local loss function $F_k$ is L-smooth and not necessarily convex. Under assumptions (A1)–(A4), if the effective local learning rate $\overline{\eta}$ (which already accounts for $\tau$ local steps and OTA scaling) is constant with $\overline{\eta} \le 1/(4L)$, then for any $T \ge 1$ the proposed coherence-aware FL scheme satisfies
\begin{align}
    \frac{1}{T}\sum_{t=1}^{T}\mathbb{E}\!\left[\|\nabla F(\btheta^{(t)})\|^2\right]
    \le \frac{4\big(F(\btheta^{(1)})-F^{*}\big)}{\overline{\eta}\,T}
    + 4\,\overline{B} + 2L\,\overline{\eta}\,\overline{\Xi}.
\end{align}
\end{theorem}

\begin{proof}
See Appendix~\ref{app:proof-nonconvex}.
\end{proof}

\section{Numerical Experiments and Discussion}
\label{sec:numerical-results}

\subsection{Simulation Setup}
Unless stated otherwise, we set $s=d$, $\rho = 10$ dB and use the power allocation calculated in Eqs.~\eqref{eq:rho_d_opt} and \eqref{eq:rho_p_opt}. The total noise used to compute the downlink SNR at static devices consists solely of receiver AWGN with variance $\sigma_w^2$ (see Section~\ref{sec:system-model}). For dynamic devices, the total noise includes both $\sigma_w^2$ and the channel estimation error introduced by product superposition, as given in Eq.~\eqref{eq:estimation-error-OFDM}. The total pilot overhead during downlink communication is denoted by $\lambda \in [0,1]$, defined as the ratio of slots used for pilot transmission (either ordinary or superposed) to the total number of downlink communication slots. The value of $\lambda$ depends on the coherence times and coherence bandwidths of dynamic links and the level of disparity among them. We denote the total number of communication rounds by $T$.

We conduct experiments using the MNIST~\cite{Lecun1998mnist} and CIFAR-10~\cite{krizhevsky2009learning} datasets.\footnote{We focus on relatively simple ML tasks here as a proof-of-concept of our innovations in addressing communication impairments in distributed learning.} For training on MNIST, we use the default convolutional neural network (CNN) architecture with convolutional and fully connected layers. For CIFAR-10, we employ the ResNet-18 architecture. 
% \textcolor{black}{For CIFAR-10 non-i.i.d. experiments, we use a standard label-skew partitioning scheme in which each device is assigned samples from a limited subset of classes, while keeping the number of samples per device fixed.} 
Each device performs local training using SGD for $\tau = 5$ local epochs with a batch size of 16. Both i.i.d. and non-i.i.d. data distributions are considered across the devices, as explained below.

\vspace{-0.1in}
\subsection{Results and Discussion}
\begin{figure}
    \centering
    \includegraphics[width=0.62\columnwidth]{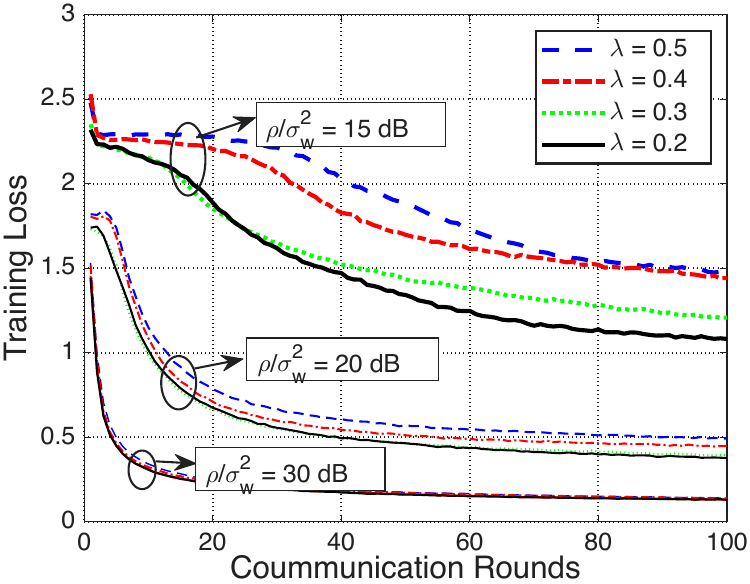}
    \caption{Training loss versus communication rounds on the MNIST dataset for the proposed product superposition-based FL scheme. \vspace{-0.2in}}
    \label{fig:train-loss}
\end{figure}
Fig.~\ref{fig:train-loss} demonstrates the training loss of FL using the proposed product superposition scheme under different SNR and $\lambda$ values for the MNIST dataset with i.i.d. distribution. Here, we assume frequency-flat channels for all links and set $M = 20$, $K = 50$, and $K_S = K_D = 25$. The results confirm that product superposition is a valid approach for FL under varying coherence disparities, which result in different $\lambda$ values. The learning performance improves significantly at higher SNRs and with lower pilot overhead, which is due to the reduced noise from dynamic devices under these conditions.

\begin{figure}[t]
    \centering

    \begin{subfigure}{\linewidth}
        \centering
        \includegraphics[width=0.62\columnwidth]{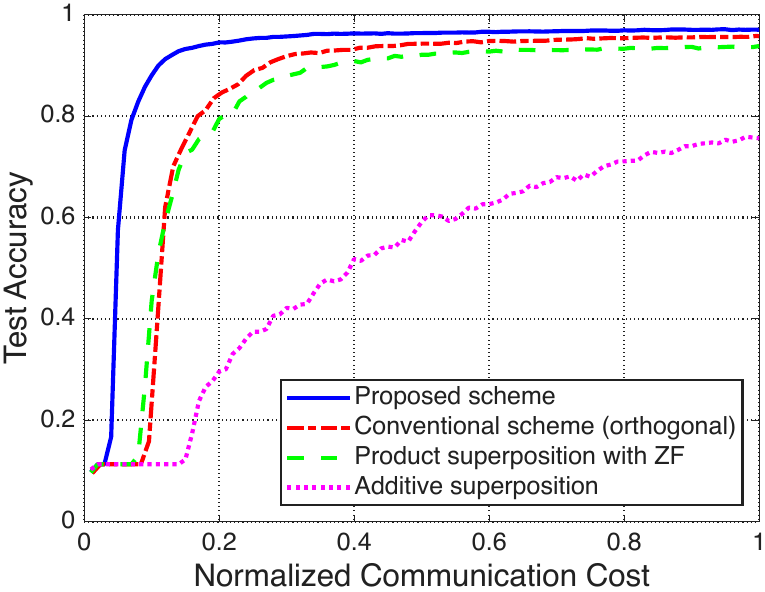}
        \caption{$\lambda = 0.2$}
        \label{fig:subfig1}
    \end{subfigure}
% \vspace{0.05cm}

    \begin{subfigure}{\linewidth}
        \centering
        \includegraphics[width=0.62\columnwidth]{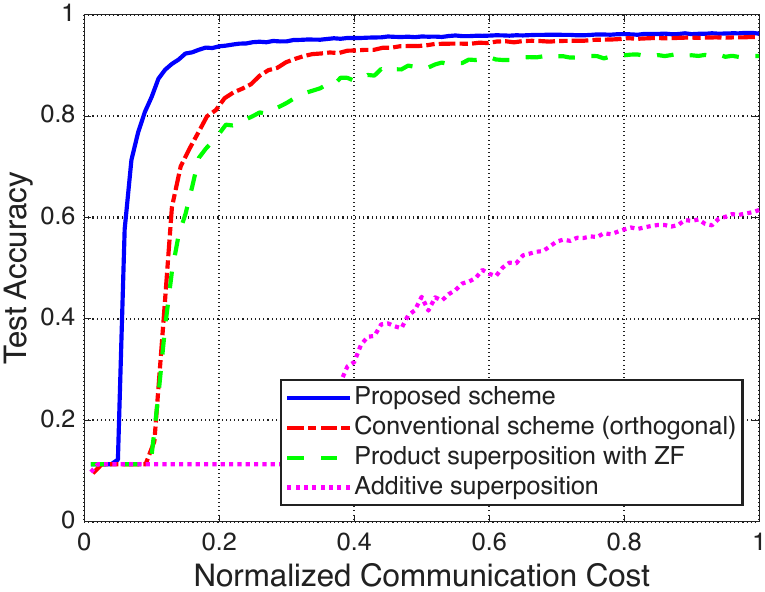}
        \caption{$\lambda = 0.3$}
        \label{fig:subfig2}
    \end{subfigure}
% \vspace{0.05cm}

    \begin{subfigure}{\linewidth}
        \centering
        \includegraphics[width=0.62\columnwidth]{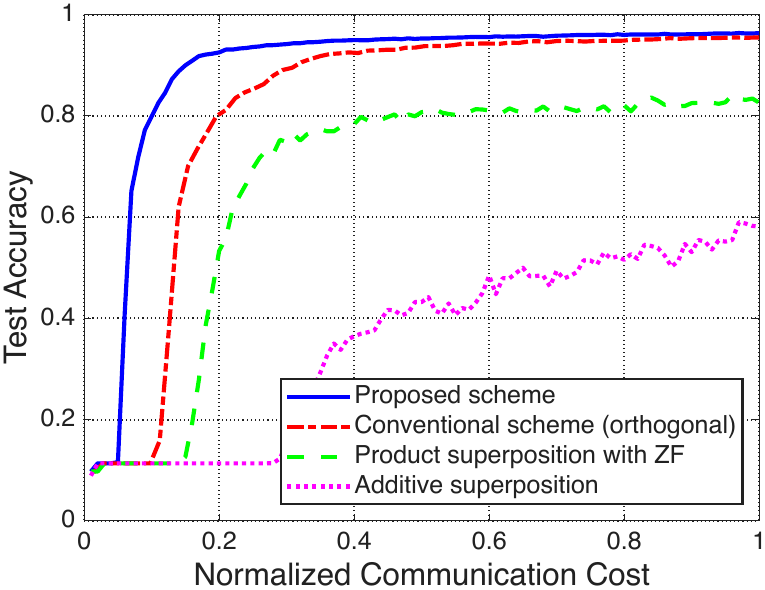}
        \caption{$\lambda = 0.4$}
        \label{fig:subfig3}
    \end{subfigure}
    \caption{Test accuracy versus normalized communication cost on the MNIST dataset for the proposed scheme and conventional baselines, under different pilot overhead values: (a) $\lambda = 0.2$, (b) $\lambda = 0.3$, and (c) $\lambda = 0.4$. \vspace{-0.2in}}
    \label{fig:snr20test-accuracy}
\end{figure}

% \begin{figure*}[t]
%     \centering

%     \begin{subfigure}[t]{0.32\textwidth}
%         \centering
%         \includegraphics[width=\linewidth]{Figures/lambda02.eps}
%         \caption{$\lambda = 0.2$}
%         \label{fig:subfig1}
%     \end{subfigure}\hfill
%     \begin{subfigure}[t]{0.32\textwidth}
%         \centering
%         \includegraphics[width=\linewidth]{Figures/lambda03.eps}
%         \caption{$\lambda = 0.3$}
%         \label{fig:subfig2}
%     \end{subfigure}\hfill
%     \begin{subfigure}[t]{0.32\textwidth}
%         \centering
%         \includegraphics[width=\linewidth]{Figures/lambda04.eps}
%         \caption{$\lambda = 0.4$}
%         \label{fig:subfig3}
%     \end{subfigure}

%     \caption{Test accuracy versus normalized communication cost on the MNIST dataset for the proposed scheme and conventional baselines, under different pilot overhead values: (a) $\lambda = 0.2$, (b) $\lambda = 0.3$, and (c) $\lambda = 0.4$. \vspace{-0.2in}}
%     \label{fig:snr20test-accuracy}
% \end{figure*}

Fig.~\ref{fig:snr20test-accuracy} shows the test accuracy versus the normalized communication cost, defined as the ratio of the total slots required for both pilot and parameter transmissions to the total slots required for parameter transmission alone, at $\lambda = \{0.2, 0.3, 0.4\}$. The plot compares the proposed product superposition scheme with benchmark methods under the MNIST dataset with i.i.d. distribution. Here, we assume frequency-flat channels for all links and set $M = 20$, \textcolor{black}{$\frac{\rho}{\sigma_w^2} = 20$ dB}, $K = 50$, $K_S = K_D = 25$, and $T = 100$. The proposed scheme with PLMF significantly outperforms conventional FL, which uses conventional signaling (orthogonal pilot and parameter transmission) for model delivery, yielding substantial gains in communication efficiency. This improvement is due to the efficient resource management enabled by product superposition—particularly the optimized pilot placement and reuse—which reduces communication overhead while maintaining high test accuracy. 

For comparison, we also include product superposition with zero-filling (ZF), where missing parameters at dynamic devices are replaced with zeros, and the additive superposition scheme, where pilot and parameter signals are simply added under coherence disparity. ZF degrades accuracy by increasing bias, while additive superposition performs poorly because the superimposed pilot interferes with parameter symbols, injecting extra noise. In contrast, our product superposition approach addresses this limitation by integrating the pilot signal into a {\em virtual channel} estimated at dynamic devices (see Eq.~\eqref{eq:ps-equivalent}) without any interference parameters. This is one of the key advantages that make the proposed method suitable for FL under coherence disparity.

Overall, the proposed scheme with PLMF outperforms all baselines. For instance, at 95\% test accuracy, it achieves a normalized communication cost reduction of approximately 0.3 compared to conventional FL.

\begin{figure}
    \centering
\includegraphics[width=0.62\columnwidth]{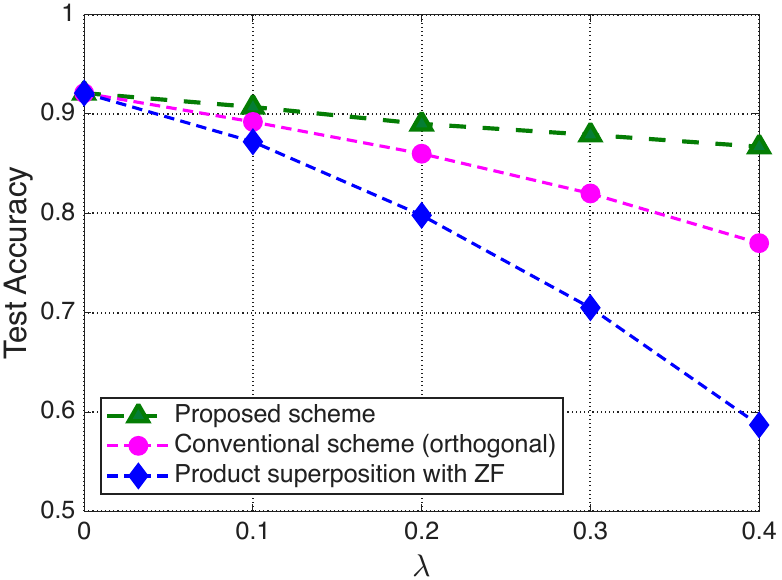}
    \caption{Test accuracy versus training overhead at fixed communication rounds $T = 20$ (MNIST dataset). \vspace{-0.2in}}
    \label{fig:testaccVSoverhead}
\end{figure}

% \begin{figure}
%     \centering
%     \includegraphics[width=\Figwidth]{Figures/snr30_miss04.eps}
%     \caption{Training loss versus normalized communication cost for the proposed product superposition-based FL scheme (MNIST).}
%     \label{fig:test-acc-snr30}
% \end{figure}

Fig. \ref{fig:testaccVSoverhead} shows the test accuracy versus the pilot overhead under the MNIST dataset with i.i.d. distribution. Here, we assume frequency-selective channels for all links and set $N=10$, $M = 20$, $\frac{\rho}{\sigma_w^2} = 20$ dB, $K = 50$, $K_S = K_D = 25$, and $T = 20$. When all devices are static, i.e., $\lambda = 0$, all schemes perform similarly. However, as $\lambda$ increases (reflecting greater coherence disparity) the performance of both the conventional scheme and the product superposition method with ZF degrades significantly. In contrast, the proposed scheme with PLMF remains robust, demonstrating its effectiveness under heterogeneous coherence conditions. At $\lambda = 0.4$, the product superposition with PLMF achieves approximately a 0.1 improvement in test accuracy over the baseline.

% Fig. \ref{fig:testaccVSoverhead} shows the test accuracy versus the pilot overhead under the same setting, with a fixed training duration of $T = 20$. When all devices are static, i.e., $\lambda = 0$, all schemes perform similarly. However, as $\lambda$ increases—reflecting greater coherence disparity—the performance of both the conventional scheme and the product superposition method with zero-filling degrades significantly. In contrast, the proposed scheme with PLMF remains robust, demonstrating its effectiveness under heterogeneous coherence conditions. At $\lambda = 0.4$, the product superposition with PLMF achieves approximately a 0.12 improvement in test accuracy over the baseline.

\begin{figure}[t]
    \centering
    \begin{subfigure}{\linewidth}
        \centering
        \includegraphics[width=0.62\columnwidth]{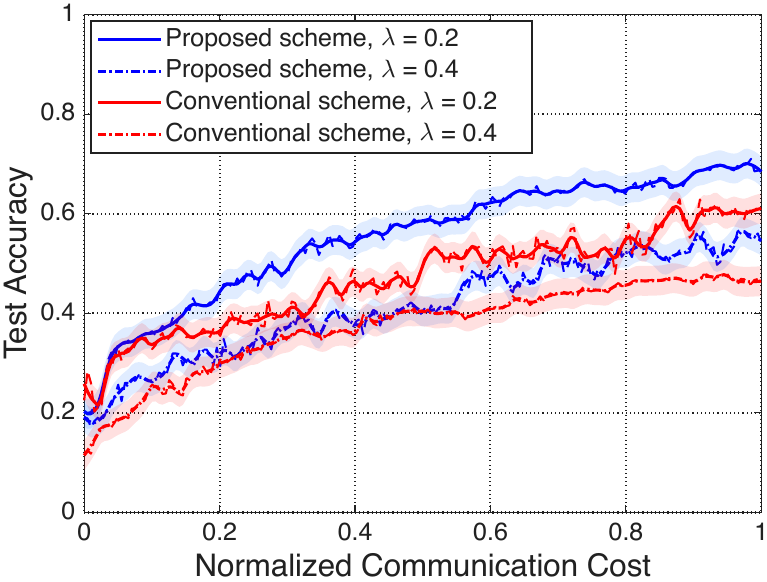}
        \caption{$\frac{\rho}{\sigma_w^2} = 20 \,\rm{dB}$ }
        \label{fig:cifar-snr10}
    \end{subfigure}
\par\vspace{0.1in}
    \begin{subfigure}{\linewidth}
        \centering
        \includegraphics[width=0.62\columnwidth]{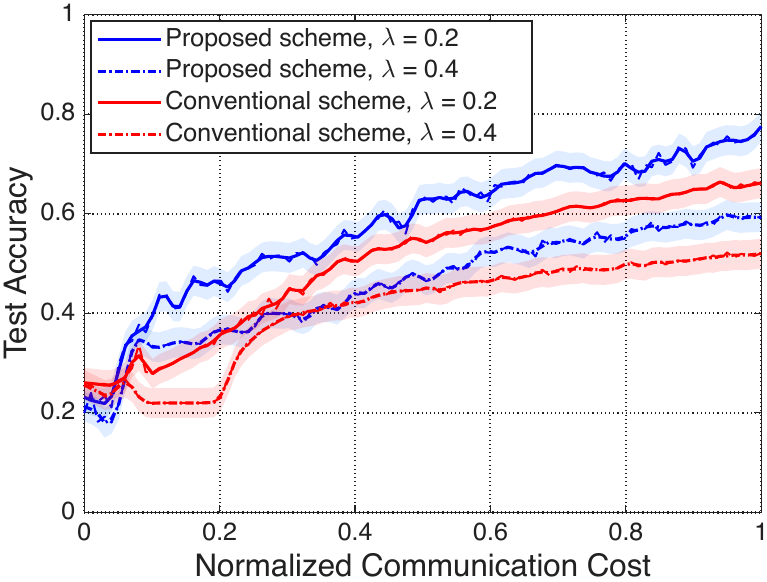}
        \caption{$\frac{\rho}{\sigma_w^2} = 30 \,\rm{dB}$ }
        \label{fig:cifar-resnet-snr30}
    \end{subfigure}    
    \caption{Test accuracy versus normalized communication cost on the CIFAR-10 dataset for the proposed scheme, and the conventional FL with ordinary pilots. Shaded regions indicate standard deviation. \vspace{-0.2in}}
    \label{fig:test-accuracy-snr30}
\end{figure}

% \vspace{-2.2cm}
% \begin{figure}
%     \centering
%     \includegraphics[width=\Figwidth]{Figures/cifar_snr10_lambda02_04.eps}
%     \caption{Test accuracy versus normalized communication cost on the CIFAR-10 dataset for the proposed product superposition-based FL, and the conventional FL with ordinary pilots.}
%     \label{fig:test-accuracy-snr30}
% \end{figure}

Fig.~\ref{fig:test-accuracy-snr30} compares the test accuracy of FL under the proposed signaling scheme and conventional signaling on the CIFAR-10 dataset across different communication rounds and pilot overheads. Here, we assume frequency-selective channels and set $N = 10$, $M = 20$, \textcolor{black}{$\frac{\rho}{\sigma_w^2} = \{20, 30\}$ dB}, $K = 60$ with $K_S = 0.6K$ and $K_D = 0.4K$, $\lambda = \{0.2, 0.4\}$, $T=100$, assuming a non-i.i.d. data distribution. The proposed product superposition scheme with the PLMF strategy consistently outperforms the other approaches, achieving significant gains in communication efficiency under coherence disparity. In particular, when \textcolor{black}{$\frac{\rho}{\sigma_w^2} = 30$ dB}, at a test accuracy of 65\%, it achieves approximately a 0.34 reduction in normalized communication cost compared to conventional FL.

\color{black}
\begin{remark}
We note that there exist several signaling-based baselines in the FL literature (e.g., digital FL); however, our work is not intended to compete with them, as it fundamentally explores a different source of gain by exploiting coherence disparity. These approaches are complementary and can be effectively combined for additive gains.
\end{remark}
\color{black}

\vspace{-0.1in}
\section{Conclusion}
\label{sec:conclusion}
This paper proposed a coherence-aware federated learning (FL) framework tailored to wireless networks in which devices exhibit heterogeneous coherence conditions. In the downlink, we employed product superposition to reuse pilot resources for global-model delivery: dynamic (short-coherence) devices estimate virtual channels, while static (long-coherence) devices receive full updates without additional spectrum or signaling. In the uplink, we incorporated coherence awareness into aggregation and mitigate partial-reception effects using previous local model filling (PLMF). We presented convergence guarantees under imperfect CSI and heterogeneous links, explicitly quantifying the roles of pilot density, power allocation, and aggregation noise. We found that the proposed design yields consistent gains in communication efficiency, latency, and learning accuracy over conventional FL baselines, while enabling efficient joint scheduling.

Taken together, these results demonstrate that explicitly accounting for coherence disparity is essential for practical FL over wireless networks. By linking air-interface mechanisms (pilot reuse, superposition structure, aggregation handling) with end-to-end learning performance, this study provides clear guidance on how to design FL systems in the presence of unequal coherence time and bandwidth, a regime pervasive in real deployments and central to AI-native 6G systems, where communication and distributed learning must be co-designed under heterogeneous channel dynamics.

% \textcolor{blue}{** advance channel training with efficient pilot placement under mismatched coherence intervals to accommodate FL (harmonizing pilot reuse and FL). future work (journal version) other aggregation models}

\appendices

\vspace{-0.1in}
\section{Pilot-Parameter Power Allocation}
\label{app:proof-PowerAllocation}
\subsection{Achievable Rate by Static Device over the Pilot Slots}
A static device $k'$ has perfect knowledge of $\bh_{k'}$, and the unitary pilot matrix, $\bX_p \in \mathbb{C}^{M \times M}$. During the pilot transmission phase (the first $M$ time slots), the signal received by user $k'$ is
\[
\by_{k',p} = \sqrt{\rho_p}\, \bh_{k'}^H \bX_{p}^\theta \bX_p + \bw_{k',p},
\]
where $\bX_{p}^\theta \in \mathbb{C}^{M \times M}$ is the parameter matrix, and $\bw_{k',p}$ is the AWGN.

To decode the parameter matrix $\bX_{p}^\theta$, the device right-multiplies the received signal by the conjugate transpose of the known pilot matrix, $\bX_p^H$. Since $\bX_p$ is unitary ($\bX_p\bX_p^H = \bI_M$), this operation effectively removes the pilot modulation
\begin{align*}
    \by'_{k',p} &= \by_{k',p} \,\bX_p^H \\
&=\sqrt{\rho_p} \,\bh_{k'}^H \bX_{p}^\theta (\bX_p \bX_p^H) + \bw_{k',p} \bX_p^H \\
&= \sqrt{\rho_p} \,\bh_{k'}^H \bX_{p}^\theta + \bw'_{k',p}.
\end{align*}
The resulting noise term, $\bw'_{k',p} \triangleq \bw_{k',p} \bX_p^H$, has the same statistical properties as the original noise. The equation above describes a standard $M \times 1$ MISO channel. The capacity, assuming i.i.d. inputs from the $M$ antennas, is given by $\mathbb{E}[\log_2(1 + \frac{\rho_p}{M\sigma_w^2} \bh_{k'}^H \bh_{k'})]$. Similarly, the rate per channel use over ($L_{t,s}-M$) data slots can be obtained as $\mathbb{E}[\log_2(1 + \frac{\rho_d}{M\sigma_w^2} \bh_{k'}^H \bh_{k'})]$. Averaging this rate over the entire block of $L_{t,s}$ symbols gives the final expression for $R_{k'}$, given in Eq.~\eqref{eq:rate_static}.

\vspace{-0.1in}
\subsection{Dynamic Device Rate and Effective SNR}
During the pilot phase, the dynamic device~$k$ estimates its {\em virtual channel}, which is defined as the product of the physical channel and the parameter matrix: $\bf_{k} = \bh_k^H \bX_\theta$. Its MMSE estimate, denoted $\overline{\bf}_{k}$, is given in Eq.~\eqref{eq:MMSE-equivalent-OFDM}, and the associated estimation error is calculated in Eq.~\eqref{eq:estimation-error-OFDM}. A key property of MMSE estimation is that the estimate $\overline{\bf}_{k}$ and the error $\tilde{\bf}_{k}$ are uncorrelated~\cite{Hassibi2003HowMuch}. Let $\alpha^2 \triangleq \frac{M\rho_p}{\sigma_w^2+M\rho_p}$. Then 
\begin{align*}
\mathbb{E}\big[||\tilde{\bf}_{k}||^2\big] &= \mathbb{E}\big[||\bf_{k}||^2\big] - \mathbb{E}\big[||\overline{\bf}_{k}||^2\big] = M(1-\alpha^2)
\end{align*}

During the data transmission phase, the received signal at a specific time slot is $\by_{k,d} = \sqrt{\rho_d}\, \bf_{k} \bX_d + \bw_k$. We substitute $\bf_{k} = \overline{\bf}_{k} + \tilde{\bf}_{k}$ to separate the signal from the effective noise
\[
\by_{k,d} = \underbrace{\sqrt{\rho_d} \,\overline{\bf}_{k} \bX_d^\theta}_{\text{signal}} + \underbrace{\sqrt{\rho_d}\, \tilde{\bf}_{k} \bX_d^\theta + \bw_k}_{\text{effective noise}}.
\]
The instantaneous SNR is the ratio of the signal power (conditioned on the estimate $\overline{\bf}_{k}$) to the variance of the effective noise. The signal power is $\rho_d ||\overline{\bf}_{k}||^2$. The noise variance is $\sigma_w^2 + \mathbb{E}[||\sqrt{\rho_d}\, \tilde{\bf}_{k} \bx_d||^2] = \sigma_w^2 + \rho_d\,\mathbb{E}[||\tilde{\bf{k}}||^2] = \sigma_w^2 + M\rho_d(1-\alpha^2)$~\cite{Hassibi2003HowMuch}. The instantaneous SNR is therefore
\begin{align*}
   \text{SNR}_k &= \frac{\rho_d ||\overline{\bf}_{k}||^2}{\sigma_w^2 + M\rho_d(1-\alpha^2)} \\
   &= \left(\frac{\rho_d(\sigma_w^2+M\rho_p)}{\sigma_w^2(\sigma_w^2+M\rho_p+M\rho_d)}\right) ||\overline{\bf}_{k}||^2. 
\end{align*}
The achievable rate for the dynamic device~$k$, $R_k$, is found by taking the expectation of $\log_2(1+\text{SNR}_k)$ over the distribution of the channel estimate $\overline{\bf}_{k}$. Therefore, the effective SNR is
\[
\gamma_{\text{eff},k} = \frac{\rho_d(\sigma_w^2+M\rho_p)}{\sigma_w^2(\sigma_w^2+M\rho_p+M\rho_d)}.
\]

\subsection{Optimal Power Allocation Derivation}

The objective is to maximize the dynamic device's rate by maximizing $\gamma_{\text{eff},k}$ subject to the total power constraint, which we assume is met with equality
\[
M \big( \rho_p + \rho_d(L_{t,s} - M) \big) = \rho \frac{s}{q} = \rho L_{t,s}.
\]
Maximizing $\gamma_{\text{eff},k}$ is equivalent to minimizing its reciprocal, $g(\rho_p, \rho_d) = \frac{\sigma_w^2}{\rho_d} + \frac{\sigma_w^2M}{\sigma_w^2+M\rho_p}$. From the power constraint, we express $\rho_p$ in terms of $\rho_d$. Let the constant $c = \frac{\rho L_{t,s}}{M}$. Then $\rho_p = c - \rho_d(L_{t,s}-M)$. Substitute this into $g(\rho_p, \rho_d)$
\[
g(\rho_d) = \frac{\sigma_w^2}{\rho_d} + \frac{\sigma_w^2M}{\sigma_w^2 + M(c - \rho_d(L_{t,s}-M))}.
\]
To find the minimum, we take the derivative with respect to $\rho_d$ and set it to zero. This results in
\[
\frac{\sigma_w^2}{\rho_d^2} = \frac{\sigma_w^2M^2(L_{t,s}-M)}{(\sigma_w^2+Mc-M(L_{t,s}-M)\rho_d)^2}.
\]
Taking the square root and rearranging to solve for $\rho_d$ yields the optimal allocation $\rho_d^*$
\begin{align*}
    \rho_d^* &= \frac{\sigma_w^2+Mc}{M\sqrt{L_{t,s}-M}(1 + \sqrt{L_{t,s}-M})} \\
    &= \frac{\sigma_w^2+\rho L_{t,s}}{M\sqrt{L_{t,s}-M}(1 + \sqrt{L_{t,s}-M})}.
\end{align*}
The optimal pilot power, $\rho_p^*$, is found by substituting $\rho_d^*$ back into the power constraint equation. This completes the proof.

% \section{Auxiliary Lemmas}
% \label{app:aux}
% \textbf{Smoothness Descent:}
% If $F$ is $L$-smooth, then for any $\btheta$ and $\mathbf{v}$,
% \begin{align}
% F(\btheta+\mathbf{v}) \le F(\btheta) + \nabla F(\btheta)^T \mathbf{v} + \frac{L}{2}\|\mathbf{v}\|^2.
% \end{align}

% \textbf{Variance Control}
% Let $\hat{\mathbf{S}}^{(t)}$ be the aggregated update and define
% \begin{align}
% \mathbf{e}_t \triangleq \hat{\mathbf{S}}^{(t)} + \bar\eta_t \nabla F(\btheta^{(t)}).
% \label{eq:def-et}
% \end{align}
% If for each coordinate $i$,
% \begin{align}
% \mathbb{E}\big[|\hat S_i - S_i^\star|^2 \mid \btheta^{(t)}\big] \le \sigma_{\mathrm{UL}}^2 + \sigma_{\mathrm{DL}}^2,
% \label{eq:coord-mse}
% \end{align}
% with $S_i^\star$ as in \eqref{eq:conv-decomp}, then
% \begin{align}
% \mathbb{E}\big[\|\mathbf{e}_t\|^2 \mid \btheta^{(t)}\big] \le \bar\eta_t^2\,\Xi_t,
% \qquad
% \Xi_t \triangleq \sigma_g^2 + d\big(\sigma_{\mathrm{UL}}^2+\sigma_{\mathrm{DL}}^2\big).
% \label{eq:et-second-moment}
% \end{align}

% \textbf{Unbiased Aggregation}
% \label{lem:unbias}
% If masks and downlink noises are independent of minibatch sampling and are zero mean, and if the OTA noise is zero mean and independent of local updates, then
% \begin{align}
% \mathbb{E}\big[\hat{\mathbf{S}}^{(t)} \mid \btheta^{(t)}\big] = -\,\bar\eta_t\,\nabla F(\btheta^{(t)}).
% \end{align}

\vspace{-0.1in}
\section{Proof of Theorem~\ref{thm:convex}}
\label{app:proof-convex}
This proof analyzes the convergence for convex and strongly convex objectives under the aggregation error model defined in the previous subsection. Let $\btheta^*$ be a minimizer of $F$.

\emph{Preliminaries (update and moments):}
The stacked PS update satisfies
\begin{align}
    \btheta^{(t+1)} = \btheta^{(t)} + \hat{S}^{(t)},
\end{align}
with the decomposition
\begin{align}
    \hat{S}^{(t)} = -\,\overline{\eta}\big(\nabla F(\btheta^{(t)}) + \mathcal{B}_t\big) + \zeta_t,
\end{align}
where $\mathbb{E}[\zeta_t \mid \mathcal{F}_t] = 0$, $\mathbb{E}[\|\zeta_t\|^2 \mid \mathcal{F}_t] \le \overline{\eta}^2 \Xi_t$, and $\mathbb{E}\|\mathcal{B}_t\|^2 \le \overline{B}$. From this,
\begin{align}
    \mathbb{E}\!\left[\|\hat{S}^{(t)}\|^2 \mid \mathcal{F}_t \right]
    &= \overline{\eta}^2 \|\nabla F(\btheta^{(t)}) + \mathcal{B}_t\|^2
    + \mathbb{E}\!\left[\|\zeta_t\|^2 \mid \mathcal{F}_t\right] \nonumber\\
    &\le 2\overline{\eta}^2 \|\nabla F(\btheta^{(t)})\|^2
    + 2\overline{\eta}^2 \|\mathcal{B}_t\|^2
    + \overline{\eta}^2 \Xi_t.
    \label{eq:sq-moment}
\end{align}

\noindent\emph{1) Convex Case}

\emph{Step 1- Descent in function value:}
By $L$-smoothness,
\begin{align}
    F(\btheta^{(t+1)})
    \le
    F(\btheta^{(t)})
    + \big\langle \nabla F(\btheta^{(t)}), \hat{S}^{(t)} \big\rangle
    + \frac{L}{2}\,\|\hat{S}^{(t)}\|^2 .
\end{align}
Taking $\mathbb{E}[\cdot \mid \mathcal{F}_t]$ and substituting $\hat{S}^{(t)}$,
\begin{align}
    \mathbb{E}\!\left[F(\btheta^{(t+1)}) \mid \mathcal{F}_t\right]
    &\le F(\btheta^{(t)}) - \overline{\eta}\,\|\nabla F(\btheta^{(t)})\|^2
     \nonumber\\
    &- \overline{\eta}\,\big\langle \nabla F(\btheta^{(t)}), \mathcal{B}_t\big\rangle \nonumber\\
    &+ \frac{L}{2}\, \mathbb{E}\!\left[\|\hat{S}^{(t)}\|^2 \mid \mathcal{F}_t\right].
\end{align}
Apply Young’s inequality $2\langle a,b\rangle \le \|a\|^2 + \|b\|^2$ on the biased inner product with parameter $L$
\begin{align}
    - \overline{\eta}\,\big\langle \nabla F(\btheta^{(t)}), \mathcal{B}_t\big\rangle
    \le \frac{\overline{\eta}}{4L}\,\|\nabla F(\btheta^{(t)})\|^2
    + \overline{\eta} L \,\|\mathcal{B}_t\|^2 .
\end{align}
Using \eqref{eq:sq-moment},
\begin{align}
    \mathbb{E}\!\left[F(\btheta^{(t+1)}) \mid \mathcal{F}_t\right]
    &\le F(\btheta^{(t)}) \nonumber\\
    &- \Big(\overline{\eta} - \frac{\overline{\eta}}{4L} - L\overline{\eta}^2\Big)\,
    \|\nabla F(\btheta^{(t)})\|^2 \nonumber\\
    &+ \big(\overline{\eta}L + L\overline{\eta}^2\big)\,\|\mathcal{B}_t\|^2
    + \frac{L}{2}\,\overline{\eta}^2 \,\Xi_t.
    \label{eq:convex-one-step}
\end{align}
With $\overline{\eta} \le 1/(4L)$ the gradient term is nonpositive, hence
\begin{align}
    \mathbb{E}\!\left[F(\btheta^{(t+1)})\right]
    \le
    \mathbb{E}\!\left[F(\btheta^{(t)})\right]
    + \overline{\eta}(L + \overline{\eta}L)\,\overline{B}
    + \frac{L}{2}\,\overline{\eta}^2 \,\overline{\Xi}.
    \label{eq:convex-progress}
\end{align}

\emph{Step 2- Telescoping via distances and averaging:}
Consider the distance recursion
\begin{align}
    \|\btheta^{(t+1)} - \btheta^*\|^2
    &= \|\btheta^{(t)} - \btheta^*\|^2
    + 2\big\langle \btheta^{(t)} - \btheta^*, \hat{S}^{(t)} \big\rangle \nonumber \\
    &+ \|\hat{S}^{(t)}\|^2 .
\end{align}
Condition on $\mathcal{F}_t$, substitute $\hat{S}^{(t)}$, use convexity
$\langle \nabla F(\btheta^{(t)}), \btheta^{(t)} - \btheta^* \rangle \ge F(\btheta^{(t)}) - F^*$,
Young’s inequality on $\langle \btheta^{(t)} - \btheta^*, \mathcal{B}_t\rangle$, and \eqref{eq:sq-moment}:
\begin{align}
    2\overline{\eta}\,\mathbb{E}\!\left[F(\btheta^{(t)}) - F^*\right]
    &\le \mathbb{E}\!\left[\|\btheta^{(t)} - \btheta^*\|^2 - \|\btheta^{(t+1)} - \btheta^*\|^2\right] \nonumber\\
    & + 4L\overline{\eta}^2\,\mathbb{E}\!\left[F(\btheta^{(t)}) - F^*\right] \nonumber \\
    &+ \overline{\eta}\,\overline{B}
    + 2\overline{\eta}^2\,\overline{B}
    + \overline{\eta}^2\,\overline{\Xi}.
\end{align}
Since $\overline{\eta}\le 1/(4L)$, we have $4L\overline{\eta}^2 \le \frac{\overline{\eta}}{2}$. Moving that term to the left:
\begin{align}
    \overline{\eta}\,\mathbb{E}\!\left[F(\btheta^{(t)}) - F^*\right]
    &\le
    \frac{1}{2}\,\mathbb{E}\!\left[\|\btheta^{(t)} - \btheta^*\|^2 - \|\btheta^{(t+1)} - \btheta^*\|^2\right] \nonumber \\
    &+ \overline{\eta}\big(2\overline{B} + L\overline{\Xi}\big).
\end{align}
Summing from $t=1$ to $T$ and using telescoping of the distance terms yields
\begin{align}
    \overline{\eta}\sum_{t=1}^{T} \mathbb{E}\!\left[F(\btheta^{(t)}) - F^*\right]
    \le
    \frac{1}{2}\,\|\btheta^{(1)} - \btheta^*\|^2
    + T\,\overline{\eta}\big(2\overline{B} + L\overline{\Xi}\big).
\end{align}
By convexity and Jensen,
\begin{align}
    \mathbb{E}\!\left[F(\overline{\btheta}^{(T)}) - F^*\right]
    &\le
    \frac{1}{T}\sum_{t=1}^{T} \mathbb{E}\!\left[F(\btheta^{(t)}) - F^*\right] \nonumber \\
    &\le
    \frac{\|\btheta^{(1)} - \btheta^*\|^2}{2\,\overline{\eta}\,T}
    + \overline{\eta}\big(L\overline{\Xi} + 2\overline{B}\big),
\end{align}
which proves the convex case.

\noindent\emph{2) $\mu$-Strongly Convex Case}

\emph{Step 1- One-step distance recursion:}
Starting from
\begin{align}
    \|\btheta^{(t+1)} - \btheta^*\|^2
    &= \|\btheta^{(t)} - \btheta^*\|^2
    + 2\big\langle \btheta^{(t)} - \btheta^*, \hat{S}^{(t)} \big\rangle \nonumber \\
    &+ \|\hat{S}^{(t)}\|^2 ,
\end{align}
take $\mathbb{E}[\cdot \mid \mathcal{F}_t]$, substitute $\hat{S}^{(t)}$, use strong convexity
$\langle \nabla F(\btheta^{(t)}), \btheta^{(t)} - \btheta^* \rangle \ge \mu \|\btheta^{(t)} - \btheta^*\|^2$,
apply Young’s inequality $2\langle a,b\rangle \le \|a\|^2 + \|b\|^2$ to the bias cross-term, and invoke~\eqref{eq:sq-moment}. We obtain
\begin{align}
\mathbb{E}\!\left[\|\btheta^{(t+1)} - \btheta^*\|^2 \mid \mathcal{F}_t\right]
    &\le \big(1 - \overline{\eta}\mu\big)\,\|\btheta^{(t)} - \btheta^*\|^2\nonumber\\
    & 
    + \overline{\eta}^2 \,\|\mathcal{B}_t\|^2
    + 2\overline{\eta}^2 \,\|\nabla F(\btheta^{(t)})\|^2 \nonumber\\
    & + 2\overline{\eta}^2 \,\|\mathcal{B}_t\|^2 
    + \overline{\eta}^2 \,\Xi_t .
\end{align}
Taking expectations and using $\mathbb{E}\|\mathcal{B}_t\|^2 \le \overline{B}$ and
$\|\nabla F(\btheta^{(t)})\|^2 \le 2L\big(F(\btheta^{(t)}) - F^*\big)$ (absorbed into a constant under smooth $+$ strongly convex), we obtain
\begin{align}
    a_{t+1} \triangleq \mathbb{E}\!\left[\|\btheta^{(t+1)} - \btheta^*\|^2\right]
    &\le \big(1 - \overline{\eta}_t\mu\big)\,a_t \nonumber \\
    &+ \overline{\eta}_t^2\,\big(2L\overline{\Xi} + 4\overline{B}\big).
\end{align}

\emph{Step 2- Solving the recursion:}
With $\overline{\eta}_t = \frac{\beta}{\mu(t+\gamma)}$ ($\beta>1$),
\begin{align}
    a_{t+1}
    \le
    \Big(1 - \frac{\beta}{t+\gamma}\Big)\,a_t
    + \frac{\beta^2}{\mu^2(t+\gamma)^2}\,C_{\mathrm{err}},
\end{align}
where $C_{\mathrm{err}} \triangleq 2L\overline{\Xi} + 4\overline{B}$. By the standard stochastic-approximation lemma,
\begin{align}
    a_{T+1} \le \frac{\nu}{T+\gamma},
\end{align}
with $\nu \triangleq \max\!\left( (\gamma+1)\,a_1, \frac{\beta^2 C_{\mathrm{err}}}{\mu^2(\beta-1)} \right)$, which proves the strongly convex case.

\vspace{-0.1in}
\section{Proof of Theorem~\ref{thm:nonconvex}}
\label{app:proof-nonconvex}
We analyze the nonconvex case using the descent lemma and aggregation-error decomposition. Let
\begin{align}
    \btheta^{(t+1)} &= \btheta^{(t)} + \hat{S}^{(t)}, \nonumber \\
    \hat{S}^{(t)} &= -\,\overline{\eta}\big(\nabla F(\btheta^{(t)})+\mathcal{B}_t\big) + \zeta_t,
\end{align}
be the stacked PS update, where $\mathbb{E}[\zeta_t\mid\mathcal{F}_t]=0$,
$\mathbb{E}[\|\zeta_t\|^2\mid\mathcal{F}_t]\le \overline{\eta}^2\,\Xi_t$,
and $\mathbb{E}\|\mathcal{B}_t\|^2\le \overline{B}$; moreover
$\overline{\Xi}\triangleq \sup_t \Xi_t$.

\emph{Step 1- Descent lemma:}
By $L$-smoothness,
\begin{align}
    F(\btheta^{(t+1)})
    \le F(\btheta^{(t)}) + \big\langle \nabla F(\btheta^{(t)}), \hat{S}^{(t)} \big\rangle
    + \frac{L}{2}\,\|\hat{S}^{(t)}\|^2 .
\end{align}
Taking conditional expectation and using linearity,
\begin{align}
    \mathbb{E}\!\left[F(\btheta^{(t+1)})\mid\mathcal{F}_t\right]
    &\le F(\btheta^{(t)}) + \big\langle \nabla F(\btheta^{(t)}),
    \mathbb{E}[\hat{S}^{(t)}\mid\mathcal{F}_t]\big\rangle \nonumber \\
    &+ \frac{L}{2}\,\mathbb{E}\!\left[\|\hat{S}^{(t)}\|^2\mid\mathcal{F}_t\right].
    \label{eq:noncvx-descent-cond}
\end{align}

\emph{Step 2- Conditional mean of the update (inner product term):}
From the update decomposition and $\mathbb{E}[\zeta_t\mid\mathcal{F}_t]=0$,
\begin{align}
    \mathbb{E}[\hat{S}^{(t)}\mid\mathcal{F}_t]
    = -\,\overline{\eta}\big(\nabla F(\btheta^{(t)})+\mathcal{B}_t\big).
\end{align}
Hence
\begin{align}
    \big\langle \nabla F(\btheta^{(t)}),
    \mathbb{E}[\hat{S}^{(t)}\mid\mathcal{F}_t]\big\rangle
    =& -\,\overline{\eta}\,\|\nabla F(\btheta^{(t)})\|^2\nonumber \\
      &-\,\overline{\eta}\,\big\langle \nabla F(\btheta^{(t)}), \mathcal{B}_t\big\rangle .
\end{align}
Apply Young’s inequality $2\langle a,b\rangle\le \|a\|^2+\|b\|^2$ with
$a=\nabla F(\btheta^{(t)})$ and $b=\mathcal{B}_t$:
\begin{align}
    -\,\overline{\eta}\,\big\langle \nabla F(\btheta^{(t)}), \mathcal{B}_t\big\rangle
    \le \frac{\overline{\eta}}{2}\,\|\nabla F(\btheta^{(t)})\|^2
      + \frac{\overline{\eta}}{2}\,\|\mathcal{B}_t\|^2 .
    \label{eq:noncvx-bias-ip}
\end{align}
Substituting into \eqref{eq:noncvx-descent-cond} gives
\begin{align}
    \mathbb{E}\!\left[F(\btheta^{(t+1)})\mid\mathcal{F}_t\right]
    &\le F(\btheta^{(t)})
    - \frac{\overline{\eta}}{2}\,\|\nabla F(\btheta^{(t)})\|^2
    + \frac{\overline{\eta}}{2}\,\|\mathcal{B}_t\|^2\nonumber \\
    &+ \frac{L}{2}\,\mathbb{E}\!\left[\|\hat{S}^{(t)}\|^2\mid\mathcal{F}_t\right].
    \label{eq:noncvx-mid1}
\end{align}

\emph{Step 3- Conditional second moment of the update:}
Condition on $\mathcal{F}_t$ and expand $\hat{S}^{(t)}=u+v$ with
$u=-\overline{\eta}(\nabla F(\btheta^{(t)})+\mathcal{B}_t)$ and $v=\zeta_t$:
\begin{align}
    \mathbb{E}\!\left[\|\hat{S}^{(t)}\|^2\mid\mathcal{F}_t\right]
    &= \|u\|^2 + 2\,\big\langle u,\mathbb{E}[v\mid\mathcal{F}_t]\big\rangle
       + \mathbb{E}\!\left[\|v\|^2\mid\mathcal{F}_t\right] \nonumber\\
    &= \overline{\eta}^2\,\|\nabla F(\btheta^{(t)})+\mathcal{B}_t\|^2
       + \mathbb{E}\!\left[\|\zeta_t\|^2\mid\mathcal{F}_t\right] \nonumber\\
    &\le 2\overline{\eta}^2\,\|\nabla F(\btheta^{(t)})\|^2
       + 2\overline{\eta}^2\,\|\mathcal{B}_t\|^2
       + \overline{\eta}^2\,\Xi_t ,
    \label{eq:noncvx-second-moment}
\end{align}
where we used $\mathbb{E}[v\mid\mathcal{F}_t]=0$ (so the cross term vanishes) and
$\|a+b\|^2\le 2\|a\|^2+2\|b\|^2$.

\emph{Step 4- Combine \eqref{eq:noncvx-mid1} and \eqref{eq:noncvx-second-moment}:}
Substituting \eqref{eq:noncvx-second-moment} into \eqref{eq:noncvx-mid1}, we have
\begin{align}
    \mathbb{E}\!\left[F(\btheta^{(t+1)})\mid\mathcal{F}_t\right]
    &\le F(\btheta^{(t)}) - \frac{\overline{\eta}}{2}\,\|\nabla F(\btheta^{(t)})\|^2
    + \frac{\overline{\eta}}{2}\,\|\mathcal{B}_t\|^2 \nonumber\\
    & + \frac{L}{2}\Big( 2\overline{\eta}^2\,\|\nabla F(\btheta^{(t)})\|^2\nonumber \\
    &+ 2\overline{\eta}^2\,\|\mathcal{B}_t\|^2
    + \overline{\eta}^2\,\Xi_t \Big).
\end{align}
Grouping by identical factors, we obtain
\begin{align}   \mathbb{E}\!\left[F(\btheta^{(t+1)})\mid\mathcal{F}_t\right]
    &\le F(\btheta^{(t)}) - \overline{\eta}\Big(\frac{1}{2}-L\overline{\eta}\Big)
       \|\nabla F(\btheta^{(t)})\|^2 \nonumber\\
    &\quad + \overline{\eta}\Big(\frac{1}{2}+L\overline{\eta}\Big)\|\mathcal{B}_t\|^2
       + \frac{L}{2}\,\overline{\eta}^2\,\Xi_t .
    \label{eq:noncvx-one-step}
\end{align}
Under $\overline{\eta}\le 1/(4L)$, we have
$\frac{1}{2}-L\overline{\eta}\ge \frac{1}{4}$ and
$\frac{1}{2}+L\overline{\eta}\le 1$.

\emph{Step 5- Telescoping and averaging:}
Take total expectation in \eqref{eq:noncvx-one-step} and rearrange
\begin{align}
    \frac{\overline{\eta}}{4}\,
    \mathbb{E}\!\left[\|\nabla F(\btheta^{(t)})\|^2\right]
    &\le \mathbb{E}\!\left[F(\btheta^{(t)})\right]
    - \mathbb{E}\!\left[F(\btheta^{(t+1)})\right] \nonumber \\
    &+ \overline{\eta}\,\mathbb{E}\!\left[\|\mathcal{B}_t\|^2\right]
    + \frac{L}{2}\,\overline{\eta}^2\,\mathbb{E}\!\left[\Xi_t\right].
\end{align}
Sum $t=1$ to $T$ and use
$\mathbb{E}\|\mathcal{B}_t\|^2\le \overline B$,
$\mathbb{E}[\Xi_t]\le \overline\Xi$, thus
\begin{align}
    \frac{\overline{\eta}}{4}\sum_{t=1}^{T}
    \mathbb{E}\!\left[\|\nabla F(\btheta^{(t)})\|^2\right]
    &\le F(\btheta^{(1)})-F^{*} \nonumber \\
    &+ T\,\overline{\eta}\,\overline{B}
    + \frac{L}{2}\,T\,\overline{\eta}^2\,\overline{\Xi}.
\end{align}
Divide by $T$ and multiply by $4/\overline{\eta}$, therefore
\begin{align}
    \frac{1}{T}\sum_{t=1}^{T}\mathbb{E}\!\left[\|\nabla F(\btheta^{(t)})\|^2\right]
    &\le \frac{4\big(F(\btheta^{(1)})-F^{*}\big)}{\overline{\eta}\,T}\nonumber \\
    &+ 4\,\overline{B} + 2L\,\overline{\eta}\,\overline{\Xi},
\end{align}
which completes the proof.

% \clearpage

% \balance
\bibliographystyle{IEEEtran}
\bibliography{IEEEabrv,ref}

\begin{IEEEbiography}
[{\includegraphics[width=1in,height=1.25in,clip,keepaspectratio]{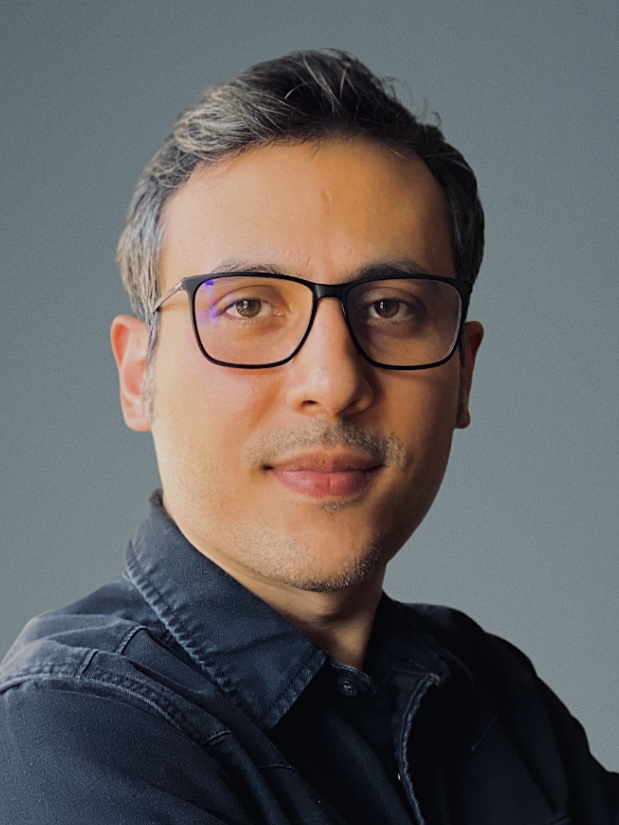}}]
{Mehdi Karbalayghareh} (Member, IEEE) received the B.Sc. degree in electrical engineering from Iran University of Science and Technology, Tehran, Iran, in 2015, the M.Sc. degree in electrical engineering from Ozyegin University, Istanbul, Turkey, in 2019, and the Ph.D. degree in electrical engineering from the University of Texas at Dallas, Richardson, TX, USA, in 2024. He is currently a Postdoctoral Researcher with the Department of Electrical and Computer Engineering at Purdue University, West Lafayette, IN, USA. Before joining Purdue University, he was conducting research at Nokia Bell Labs in Murray Hill, NJ. His research interests include information theory, machine learning, and wireless communications. He was the recipient of the 2023 Excellence in Education Doctoral Fellowship and the 2024 Research Excellence Award at the University of Texas at Dallas. He was also awarded the IEEE ComSoc Student Grant for IEEE ICC 2024.
\end{IEEEbiography}

\begin{IEEEbiography}
[{\includegraphics[width=1in,height=1.25in,clip,keepaspectratio]{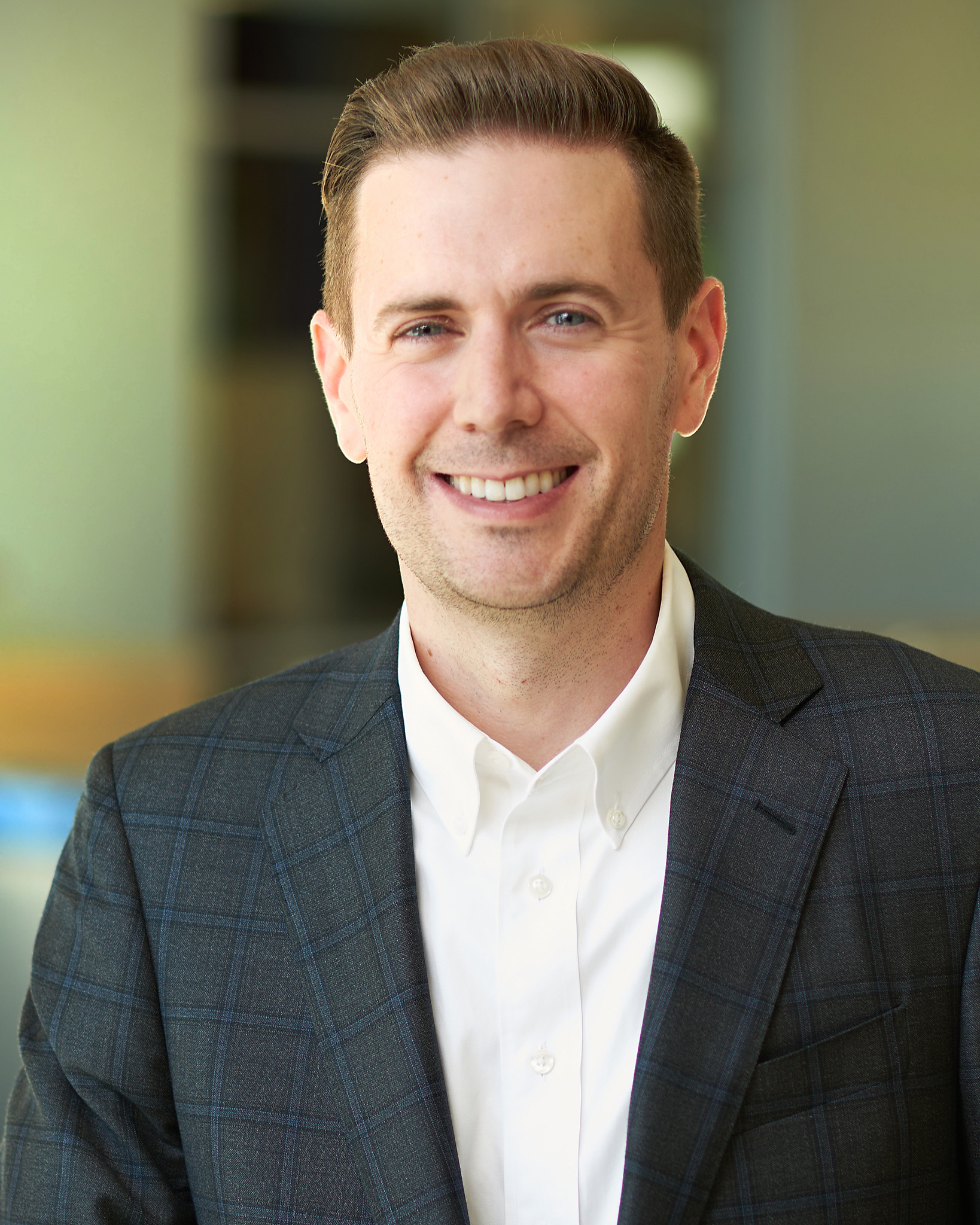}}]
{David J. Love} (Fellow, IEEE) received the B.S. (with highest honors), M.S.E., and Ph.D. degrees in electrical engineering from the University of Texas at Austin in 2000, 2002, and 2004, respectively. Since 2004, he has been with the Elmore Family School of Electrical and Computer Engineering at Purdue University, where he is now the Nick Trbovich Professor of Electrical and Computer Engineering. He serves as a Senior Editor for IEEE Journal on Selected Areas in Communications (JSAC) and previously served as a Senior Editor for IEEE Signal Processing Magazine, Editor for the IEEE Transactions on Communications, Associate Editor for the IEEE Transactions on Signal Processing, and guest editor for special issues of the JSAC and the EURASIP Journal on Wireless Communications and Networking. He was a member of the Executive Committee for the National Spectrum Consortium. He holds 32 issued U.S. patents. His research interests are in the design and analysis of broadband wireless communication systems, beyond-5G wireless systems, multiple-input multiple-output (MIMO) communications, millimeter wave wireless, software defined radios and wireless networks, coding theory, and MIMO array processing. He is a Fellow of the American Association for the Advancement of Science (AAAS) and the National Academy of Inventors (NAI). He was named a Thomson Reuters Highly Cited Researcher (2014 and 2015). He received the 2025 IEEE Communications Society Signal Processing and Computing for Communications (SPCC) Technical Recognition Award. Along with his co-authors, he won best paper awards from the IEEE Communications Society (2016 Stephen O. Rice Prize, 2020 Fred W. Ellersick Prize, and 2024 William R. Bennett Prize), the IEEE Signal Processing Society (2015 IEEE Signal Processing Society Best Paper Award), and the IEEE Vehicular Technology Society (2010 Jack Neubauer Memorial Award).
\end{IEEEbiography}
\vfill\break
\begin{IEEEbiography}
[{\includegraphics[width=1in,height=1.25in,clip,keepaspectratio]{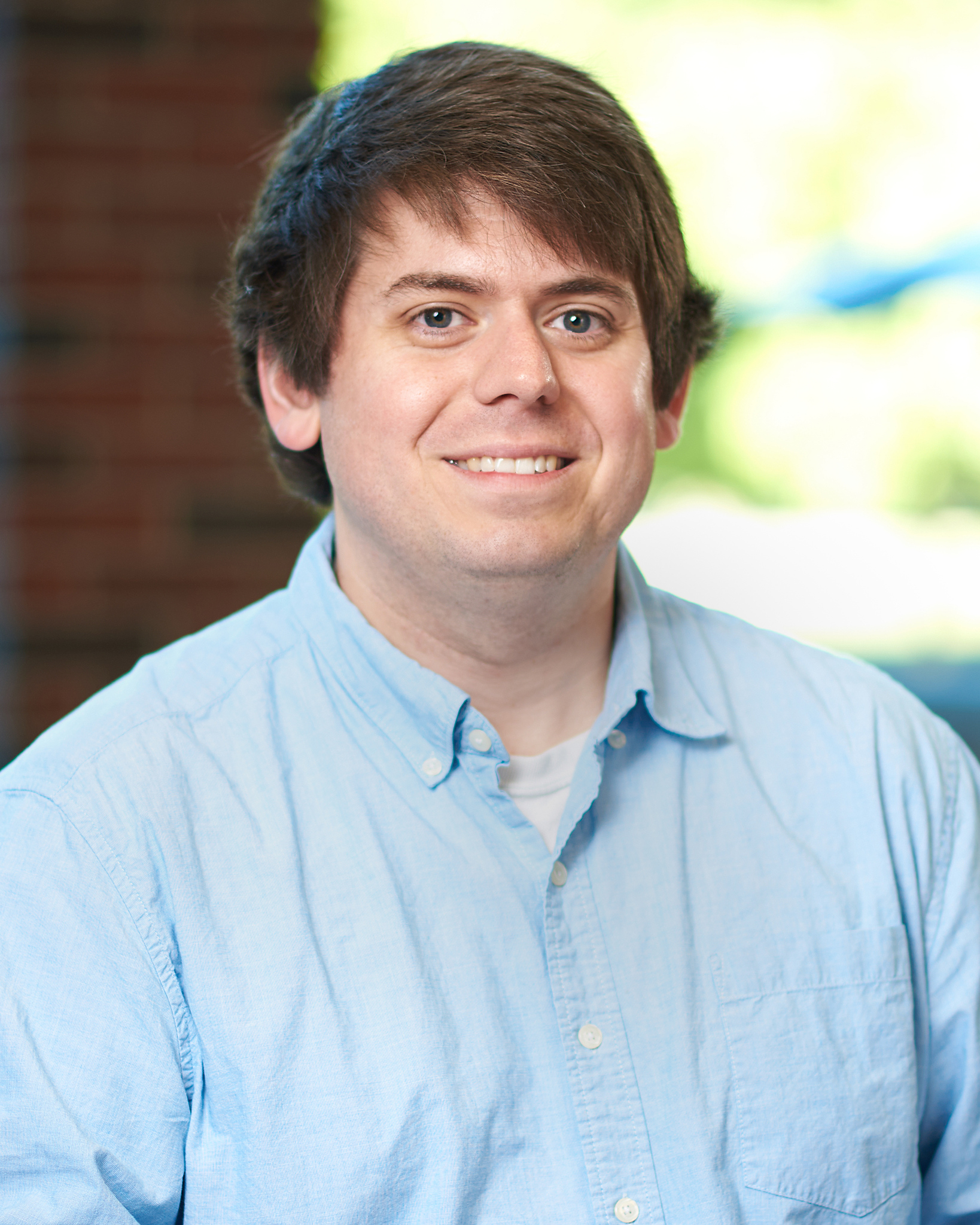}}]
{Christopher G. Brinton} (Senior Member, IEEE) is the Elmore Associate Professor of Electrical and Computer Engineering (ECE) and Faculty Director of the HIVE Engineering Entrepreneurship Center at Purdue University. His research interest is at the intersection of machine learning, communications, and networking, specifically in fog/edge network intelligence, distributed machine learning, and AI/ML-inspired wireless network optimization. Dr. Brinton is a recipient of five of the US top early career awards, from the National Science Foundation (CAREER), Office of Naval Research (YIP), Defense Advanced Research Projects Agency (YFA and Director’s Fellowship), and Air Force Office of Scientific Research (YIP). He is also a recipient of the IEEE Communication Society William Bennett Prize Best Paper Award, the Intel Rising Star Faculty Award, the Qualcomm Faculty Award, and Purdue College of Engineering Faculty Excellence Awards in Early Career Research, Early Career Teaching, and Online Learning. Dr. Brinton currently serves as an Associate Editor for IEEE/ACM Transactions on Networking, and previously was an Associate Editor for IEEE Transactions on Wireless Communications. Prior to joining Purdue, Dr. Brinton was the Associate Director of the EDGE Lab and a Lecturer of Electrical Engineering at Princeton University. He also co-founded Zoomi Inc., a big data startup company that has provided learning optimization to more than one million users worldwide and holds US Patents in machine learning for education. His book The Power of Networks: 6 Principles That Connect our Lives and associated Massive Open Online Courses (MOOCs) have reached over 400,000 students. Dr. Brinton received the PhD (with honors) and MS Degrees from Princeton in 2016 and 2013, respectively, both in Electrical Engineering.
\end{IEEEbiography}

% \balance

\end{document}